\documentclass[a4paper,UKenglish,cleveref, autoref, thm-restate]{lipics-v2021}

\title{Recognition of Linear and Star Variants of Leaf Powers is in P} 

\author{Benjamin Bergougnoux}{University of Bergen, Norway}{Benjamin.Bergougnoux@uib.no}{}{}
\author{Svein H\o gemo}{University of Bergen, Norway}{Svein.Hogemo@uib.no}{}{}
\author{Jan Arne Telle}{University of Bergen, Norway}{Jan.Arne.Telle@uib.no}{}{}
\author{Martin Vatshelle}{University of Bergen, Norway}{Martin.Vatshelle@uib.no}{}{}

\authorrunning{B. Bergougnoux, S. H\o gemo, J. A. Telle and M. Vatshelle} 

\Copyright{Benjamin Bergougnoux, Svein H\o gemo, Jan Arne Telle and Martin Vatshelle} 



\category{} 

\relatedversion{} 

\supplement{}

\acknowledgements{}

\nolinenumbers 

\hideLIPIcs  

\usepackage[ruled,vlined,linesnumbered]{algorithm2e}
\usepackage{boxedminipage,verbatim,float,caption}
\usepackage{framed}
\newtheorem{fact}[theorem]{Fact}

\def\bQ{\mathbb{Q}}

\def\cA{\mathcal{A}}
\def\cB{\mathcal{B}}

\def\cF{\mathcal{F}}

\def\cL{\mathcal{L}}

\def\cI{\mathcal{I}}

\def\cN{\mathcal{N}}

\def\cX{\mathcal{X}}

\def\cT{\mathcal{T}}

\mathchardef\mhyphen="2D

\def\abs#1{\lvert #1\rvert}
\newcommand{\incomp}{<\kern-1.1em{>}}

\newcommand{\notmin}{\mathsf{notmin}}

\newcommand{\new}{\mathsf{new}}

\def\maxcap{X\text{-}\max_{\cap}}
\def\mincap{X\text{-}\min_{\cap}}

\newcommand{\cc}{\mathsf{CC}}

\newcommand{\w}{\mathsf{w}}
\renewcommand{\min}{\mathsf{min}}
\renewcommand{\max}{\mathsf{max}}
\renewcommand{\emptyset}{\varnothing}
\renewcommand{\leq}{\leqslant}
\renewcommand{\geq}{\geqslant}
\renewcommand{\mid}{\, : \, }

\begin{document}

\maketitle  

\begin{abstract}

A $k$-leaf power of a tree $T$  is a graph $G$  whose vertices are the leaves of  $T$ and whose edges connect pairs of leaves whose distance in $T$ is at most $k$. A graph is a leaf power if it is a $k$-leaf power for some $k$.
        Over 20 years ago, Nishimura et al. [J. Algorithms, 2002] asked if recognition of leaf powers was in P. Recently, Lafond [SODA 2022] showed an XP algorithm when parameterized by $k$, while leaving the main question open. In this paper, we explore this question from the perspective of two alternative models of leaf powers, showing 
        that both a linear and a star variant of leaf powers can be recognized in polynomial-time.
        
        

\end{abstract}

\section{Introduction}

Leaf powers were introduced by Nishimura et al. in \cite{NishimuraRT02}, and have enjoyed a steady stream of research. Leaf powers are related to the problem of reconstructing \emph{phylogenetic trees}. 
For an integer $k$, a graph $G$ is a $k$-leaf power if there exists a tree $T$ -- called a \emph{leaf root} -- with a one-to-one correspondence between $V(G)$ and the leaves of $T$, such that two vertices $u$ and $v$ are neighbors in $G$ iff the distance between the two corresponding leaves in $T$ is at most $k$. $G$ is a leaf powers if it is a $k$-leaf power for some  
$k$.
The most important open problem in the field is whether leaf powers can be recognized in polynomial time. 

Most of the results on leaf powers have followed two main lines, focusing either on the distance values $k$ or on the relation of leaf powers to other graph classes, see e.g. the survey by Calamoneri et al.~\cite{CalamoneriS16}. 
For the first approach, steady research for increasing values of $k$ has shown that $k$-leaf powers for any $k\leq 6$ is recognizable in polytime~\cite{BrandstadtL06,BrandstadtLS08,ChangK07,DomGHN05,Ducoffe19,NishimuraRT02}. Moreover, the recognition of $k$-leaf powers is known to be FPT parameterized by $k$ and the degeneracy of the graph~\cite{EppsteinH20}. Recently Lafond~\cite{Lafond22} gave a polynomial time algorithm to recognize $k$-leaf powers for any constant value of $k$.
For the second approach, we can mention that interval graphs~\cite{BrandstadtLS08} and rooted directed path graphs~\cite{BrandstadtHMW10} are leaf powers, and also that leaf powers have mim-width one~\cite{JaffkeKST19} and are strongly chordal. 
Moreover, an infinite family of strongly chordal graphs that are not leaf powers has been identified~\cite{Lafond17}; see also Nevris and Rosenke~\cite{NevriesR16}.

To decide if leaf powers are recognizable in polynomial time, it may be better not to focus on the distance values $k$. Firstly, the specialized algorithms for $k$-leaf powers for $k \leq 6$ do not seem to generalize. Secondly, the recent XP algorithm of Lafond~\cite{Lafond22} uses techniques that will not allow removing $k$ from the exponent.
%
In this paper we therefore take a different approach, and consider alternative models for leaf powers that do not rely on a distance bound. In order to make progress towards settling the main question, we consider fundamental restrictions on the shape of the trees, in two distinct directions: subdivided caterpillars (linear) and subdivided stars, in both cases showing polynomial-time recognizability. We use two models: weighted leaf roots for the linear case and NeS models for the stars. 



The first model uses rational edge weights between 0 and 1 in the tree $T$ which allows to fix a bound of 1 for the tree distance. It is not hard to see that this coincides with the standard definition of leaf powers using an unweighted tree $T$ and a bound $k$ on distance. Given a solution of the latter type we simply set all edge weights to $1/k$, while in the other direction we let $k$ be the least common denominator of all edge weights and then subdivide each edge a number of times equal to its weight times $k$. 

The second model arises by combining the result of Brandstädt et~al. that leaf powers are exactly the fixed tolerance NeST graphs~\cite[Theorem 4]{BrandstadtHMW10} with the result of
Bibelnieks et~al.~\cite[Theorem 3.3]{BibelnieksD93} that these latter graphs are exactly those that admit what they call a ``neighborhood subtree intersection representation'', that we choose to call a NeS model. 
NeS models are a generalization of interval models: 
by considering intervals of the line as having a center that stretches uniformly in both directions, we can generalize the line to a tree embedded in the plane, and the intervals to embedded subtrees with a center, stretching uniformly in all directions from the center, along tree edges. Thus a NeS model of a graph $G$ consists of an embedded tree and one such subtree for each vertex, such that two vertices are adjacent in $G$ iff their subtrees have non-empty intersection. Precise definitions are given later. The leaf powers are exactly the graphs having a NeS model. Some results are much easier to prove using NeS models, to illustrate this, we show that leaf powers are closed under several operations such as the addition of a universal vertex (see \cref{lem:reductionrules}). 

We show that fundamental constraints on these models allow polynomial-time recognition.
Using the first model, we restrict to edge-weighted caterpillars, i.e. trees with a path containing all nodes of degree 2 or more. 
We call linear leaf power a graph with such model.
Brandst\"adt et al.~\cite{BrandstadtH08} considered leaf roots restricted to caterpillars (see also \cite{CalamoneriFS14}) in the unweighted setting, showing that unit interval graphs are exactly the  $k$-leaf powers for some $k$ with a leaf root being an unweighted caterpillar. 
In the unweighted setting, linear leaf powers are graphs with a leaf root that is a {\em subdivision} of a caterpillar.
We show that linear leaf powers are exactly the \emph{co-threshold tolerance graphs} \cite{MonmaRT88}, and combined with the algorithm of Golovach et al.~\cite{GolovachHLMSSS17} this implies that we can recognize linear leaf powers in $O(n^2)$ time. Our proof goes via the equivalent concept of blue-red interval graphs introduced by \cite{GolumbicWL14}, see Figure~\ref{fig:fullexample}.

The recognition of linear leaf powers in polynomial time could have practical applications
for deciding whether the most-likely evolutionary tree associated with a set of organisms has a linear topology.
Answering this question might find particular relevance inside the field of Tumor phylogenetics where,  under certain model assumptions, linear topologies are considered more likely~\cite{Tumor2,Tumor1}.

For NeS models, we restrict to graphs having a NeS model where the embedding tree is a star, and show that they can be recognized in polynomial time. Note that allowing the embedding tree to be a subdivided star will result in the same class of graphs. Our algorithm uses the fact that the input graph must be a chordal graph, and for each maximal clique $X$ we check if $G$ admits a star NeS model where the set of vertices having a subtree containing the central vertex of the star is $X$.
To check this we use a combinatorial characterization, that we call a ``good partition'', of a star NeS model.


\section{Preliminaries}

For positive integer $k$, denote by $[k]$ the set $\{1,2,\dots,k\}$.
A partition of a set $S$ is a collection of non-empty disjoint subsets $B_1,\dots,B_t$  of $S$ -- called \textit{blocks} -- such that $S=B_1\cup\dots\cup B_t$.
Given two partitions $\cA,\cB$ of $S$, we say  $\cA\sqsubseteq \cB$ if every block of $\cA$ is included in a block of $\cB$, i.e. $\sqsubseteq$ is the \emph{refinement} relation.

\noindent{\bf Graph.}
Our graph terminology is standard and we refer to \cite{Diestel12}. 
The vertex set of a graph $G$ is denoted by $V(G)$ and its edge set by $E(G)$. 
An edge between two vertices $x$ and $y$ is denoted by $xy$ or $yx$. 
The set of vertices that is adjacent to $x$ is denoted by $N(x)$.
A vertex $x$ is simplicial if $N(x)$ is a clique.
Two vertices $x,y$ are true twins if $xy\in E(G)$ and $N(x)\setminus\{y\}=N(y)\setminus\{x\}$.
Given $X\subseteq V(G)$, we denote by $G[X]$ the graph induced by $X$.
Given a vertex $v\in V(G)$, we denote the subgraph $G[V(G)\setminus \{v\}]$ by $G-v$.
We denote by $\cc(G)$ the partition of $G$ into its connected components.
Given a tree $T$ and an edge-weight function $\w : E(T) \to \bQ$, the distance between two vertices $x$ and $y$ is denoted by $d_T(x,y)$ is $\sum_{e\in E(P)}\w(e)$ with $P$ is the unique path between $x$ and~$y$.

\noindent{\bf Leaf power.}	
In the Introduction we have already given the standard definition of leaf powers and leaf roots, and also we argued the equivalence with the following. Given a graph $G$, a leaf root of $G$ is a pair $(T,\w)$ of a  tree $T$ and a rational-valued weight function $\w : E(T) \to [0,1]$ such that the vertices of $G$ are the leaves of $T$ and for every $u,v\in V(G)$, $u$ and $v$ are adjacent iff $d_T(u,v)\leq 1$.
Moreover, if $T$ is a caterpillar we call $(T,\w)$ a \emph{linear leaf root}.
A graph is a \emph{leaf power} if it admits a leaf root and it is a linear leaf power if it admits a linear leaf root. 
Since we manipulate both the graphs and the trees representing them, the vertices of trees will be called \emph{nodes} to avoid confusion.

\noindent{\bf Interval graphs.}
A graph	$G$ is an \emph{interval graph} if there exists a set of intervals in $\bQ$, $\cI = (I_v)_{v\in V(G)}$, such that for every pair of vertices $u,v\in V(G)$, the intervals $I_v$ and $I_u$ intersect iff $uv\in E(G)$.
We call $(I_v)_{v\in V(G)}$ an \emph{interval model} of $G$.
For an interval $I=[\ell,r]$,
we define the midpoint of $I$ as $(\ell+r)/2$ and its length as $r-\ell$.

\noindent{\bf Clique tree.} For a chordal graph $G$, a clique tree $CT$ of $G$ is a tree whose vertices are the maximal cliques of $V(G)$ and for every vertex $v\in V(G)$, the set of maximal cliques of $G$ containing $v$ induces a subtree of $CT$.
Figure~\ref{fig:CTandNES} gives an example of clique tree.
Every chordal graph admits $O(n)$ maximal cliques and given a graph $G$, in time $O(n+m)$ we can construct a clique tree of $G$ or confirm that $G$ is not chordal~\cite{HabibMPV00,RoseTL76}.
When a clique tree is a path, we call it a clique path.
We denote by $(K_1,\dots,K_k)$ the clique path whose vertices are $K_1,\dots,K_k$ and where $K_i$ is adjacent to $K_{i+1}$ for every $i\in [k-1]$.

\section{Linear leaf powers}\label{sec:equivdef}

In this section we show that linear leaf powers are exactly the \emph{co-threshold tolerance graphs} (co-TT graphs). Combined with the algorithm in~\cite{GolovachHLMSSS17}, this implies that we can recognize linear leaf powers in $O(n^2)$ time.


Co-TT graphs were defined by Monma, Trotter and Reed in \cite{MonmaRT88}; we will not define them here as we do not use this characterization.
Rather, we work with the equivalent class of \emph{blue-red interval graphs}~\cite[Proposition 3.3]{GolumbicWL14}. 

\begin{definition}[Blue-red interval graph]\label{def:redblue}
	A graph $G$ is a \emph{blue-red interval} graph if there exists a bipartition $(B,R)$ of $V(G)$ and an interval model $\cI=(I_v)_{v\in V(G)}$ (with $(B,R,\cI)$ called a blue-red interval model) such that 
	$E(G) = \{b_1b_2 \mid b_1,b_2\in B \text{ and } I_{b_1}\cap I_{b_2} \neq \emptyset \} \cup \{rb \mid r\in R , b\in B \text{ and } I_r \subseteq I_b\}.$
\end{definition}
The red vertices induce an independent set, $(I_b)_{b\in B}$ is an interval model of $G[B]$, and we have a blue-red edge for each red interval contained in a blue interval.
The following fact can be easily deduced from Figure \ref{fig:intersectioninclusion}.
\begin{fact}\label{fact:intervals}
	Consider two intervals $I_1,I_2$ with lengths $\ell_1,\ell_2$ and midpoints $m_1,m_2$ respectively.
	We have $I_1\cap I_2\neq \emptyset$ iff $|m_1-m_2|\leq \frac{\ell_1+\ell_2}{2}$.
	Moreover, we have $I_2\subseteq I_1$ iff $|m_1-m_2| \leq \frac{\ell_1 - \ell_2}{2}$.
\end{fact}
\begin{figure}[ht]
	\centering
	\includegraphics[width=\linewidth]{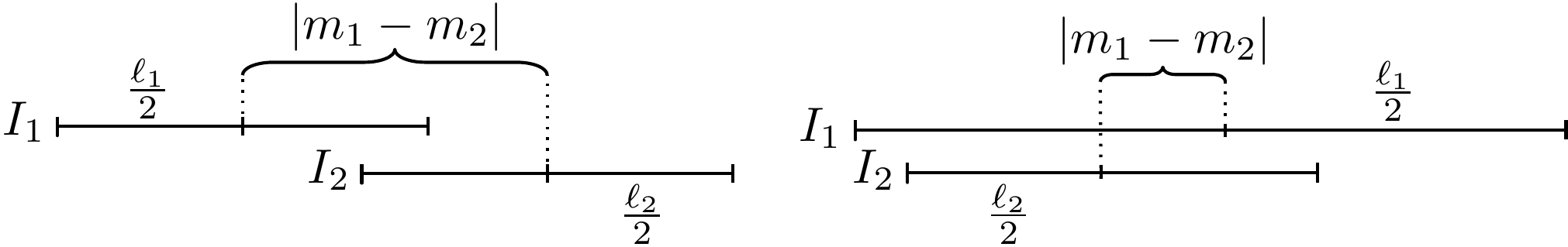}
	\caption{Example of two intervals overlapping and one interval containing another one.}
	\label{fig:intersectioninclusion}
\end{figure}
%
%
To prove that linear leaf powers are exactly blue-red interval graphs, we use a similar construction as the one used in~\cite[Theorem 6]{BrandstadtH08} to prove that every interval graph is a leaf power, but in our setting, we have to deal with red vertices and this complicates things quite a bit.
\begin{restatable}[$\star$]{theorem}{equivalence}
	    \label{thm:equivalence}
			$G$ is a blue-red interval graph iff $G$ is a linear leaf power.
\end{restatable}
\begin{proof}
	($\Rightarrow$) Let $G$ be a blue-red interval graph with blue-red interval model $(B,R,(I_v)_{v\in V(G)})$.
	We assume w.l.o.g. that $G$ is connected as otherwise we can obtain a leaf root of $G$ from leaf roots of its connected components by creating a new node adjacent to an internal node of each leaf root via edges of weight 1.
	For every $v\in V(G)$, we denote by $\ell(v)$ and  $m(v)$ the length and the midpoint of the interval $I_{v}$.
	We suppose w.l.o.g. that, for every $v\in V(G)$, we have $0<\ell(v)\leq 1$ as we can always divide the endpoints of all intervals by $\max\{\ell(v)\mid v\in V(G)\}$ and add some $\epsilon>0$ to the right endpoints of the intervals of length 0.
	We fix an ordering on $V(G)$, $\{v_1,\dots,v_n\}$, such that $m(v_i) \leq m(v_j)$ for every $i<j$.
	
	We define $T$ a caterpillar and its edges $f_1,\dots,f_n,e_1,\dots,e_{n-1}$ as depicted in Figure~\ref{fig:linearlayout}.
	Let $\w : E(T)\to [0,1]$ such that for every $i\in [n-1]$ we have $\w(e_i)=m(v_{i+1})-m(v_i)$ and for every $i\in [n]$ the weight of $f_i$ is $\frac{1-\ell(v_i)}{2}$ if $v_i\in B$ and $\frac{1+\ell(v_i)}{2}$ if $v_i\in R$.
	Since we assume that $G$ is connected and the length of the intervals are at most 1, for every $i\in [n-1]$, we have $\w(e_i)=m(v_{i+1})-m(v_i) \leq 1$ so the weights are well-defined.
	
	\begin{figure}[ht]
		\centering
		\includegraphics[width=0.7\linewidth]{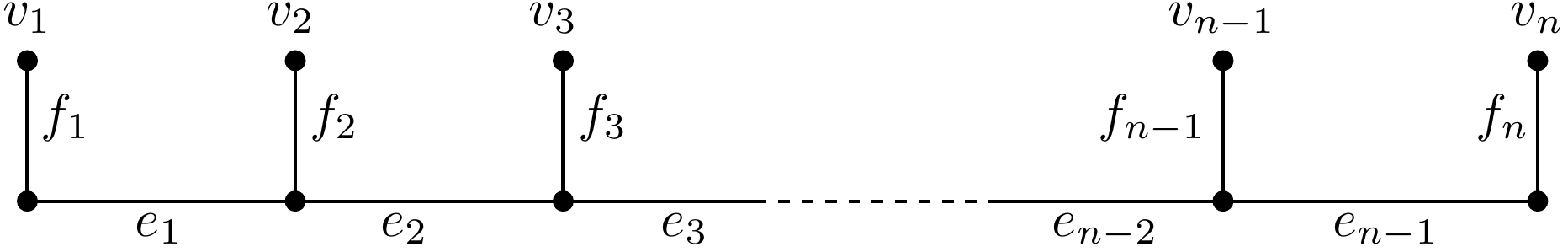}
		\caption{The linear branch-decomposition used to prove that every blue-red interval graph is a linear leaf power. Observe that the weight of an edge $e_i$ can be 0 if $I_{v_i}$ and $I_{v_{i-1}}$ have the same midpoint, in that case, $e_i$ can be contracted.}
		\label{fig:linearlayout}
	\end{figure}
	
	We claim that $(T,\w)$ is a linear leaf root of $G$.
	Let $i,j\in[n]$ such that $i<j$.
	We have to prove that $v_iv_j \in E(G)$ iff $d_T(v_i,v_j)\leq 1$.
	The edges of the path between $v_i$ and $v_j$ in $T$ are $f_i,e_i,e_{i+1},\dots,e_{j-1},f_j$. By construction, we have 
	{\small \begin{equation*}
			\w(e_i)+\dots+\w(e_{j-1}) = m_{i+1}-m_i + m_{i+2} - m_{i+1} +\dots + m_{j-1} - m_{j-2} + m_j - m_{j-1} = m_j - m_i.
	\end{equation*}}
	Hence $d_T(v_i,v_j) = \w(f_i) + |m_j - m_i| + \w(f_j)$. 
	If $v_i$ and $v_j$ are red vertices, then $\w(f_i),\w(f_j)>\frac{1}{2}$ because we suppose that $\ell(x)>0$ for every vertex $x$.
	Thus, $d_T(v_i,v_j)>1$ for every pair of red vertices $(v_i,v_j)$.
	So at least one vertex among $v_i,v_j$ is blue.
	Suppose w.l.o.g. that $v_i$ is blue.
	As $d_T(v_i,v_j) = \w(f_i) + |m_j - m_i| + \w(f_j)$, we deduce that
	\begin{equation}\label{eq:distance}
		d_T(v_i,v_j)=\begin{cases} 1 + |m(v_i) - m(v_j)| -\left( \frac{\ell(v_i)+\ell(v_j)}{2} \right) &\text{ if } v_j\in B\\
			1 + |m(v_i) - m(u_i)| -\left( \frac{\ell(v_i)-\ell(v_j)}{2} \right) &\text{ if } v_j\in R.
		\end{cases} 
	\end{equation}
	If $v_j\in B$ then by Fact \ref{fact:intervals}, we have $I_{v_i}\cap I_{v_j}\neq \emptyset$ iff $|m(v_i) - m(v_j)| \leq \frac{\ell(v_i)+\ell(v_j)}{2}$.
	We conclude that $v_iv_j\in E(G)$ iff $d_T(v_i,v_j)= |m(v_i) - m(v_j)| -\left( \frac{\ell(v_i)+\ell(v_j)}{2} \right) + 1 \leq 1.$
	On the other hand, if $v_j\in R$ then $I_{v_j}\subseteq I_{v_i}$ iff $|m(v_i) - m(v_j)| \leq \frac{\ell(v_i)-\ell(v_j)}{2}$.
	We conclude that $v_iv_j\in E(G)$ iff $d_T(v_i,v_j)= |m(v_i) - m(v_j)| -\left( \frac{\ell(v_i)-\ell(v_j)}{2} \right) + 1 \leq 1.$
	Hence, $(T,\w)$ is a linear leaf root of $G$. This proves that every blue-red interval graph is a linear leaf power.
	
	\bigskip
	($\Leftarrow$) Let $G$ be a linear leaf power and  $(T,\w)$ a linear leaf root of $G$.
	Let $(u_1,\dots,u_t)$ be the path induced by the internal vertices of $T$.
	We suppose w.l.o.g. that $G$ does not contain isolated vertices as we can easily deal with such vertices by associating each of them with an interval that does not intersect the other intervals.
	Consequently, for every leaf $a$ of $T$ adjacent to some $u_i$, we have $d_T(a,u_i)\leq 1$.
	For every vertex $v\in V(G)$ whose neighbor in $T$ is $u_i$, we associate $v$ with an interval $I_v$ and a color such that the midpoint of $I_v$ is $m(v)= d_T(u_{1},u_{i})$ and
	\begin{itemize}
		\item if $\w(u_iv)\leq \frac{1}{2}$, the length of $I_v$ is $\ell(v)=1-2\w(u_iv)$ and the color of $v$ is blue,
		\item otherwise ($\frac{1}{2} < \w(u_iv)\leq 1$), the length of $I_v$ is $\ell(v)=2\w(u_iv)-1$ and the color of $v$ is red.
	\end{itemize}
	Let $B,R\subseteq V(G)$ be the sets of blue and red vertices respectively.
	Observe that the red vertices induce an independent set since their distance to the inner path is strictly more than $1/2$.
	By construction, we have the following equation for every vertex $v$ whose neighbor in $T$ is $u_i$
	\begin{equation}\label{eq:distance1}
		w(u_iv)=\begin{cases}
			\frac{1-\ell(v)}{2}  &\text{if } v\in B,\\
			\frac{1+\ell(v)}{2} & \text{if } v\in R.	
		\end{cases} 
	\end{equation}		
	Moreover, from the definition of the midpoint's, we deduce that for every $v_i,v_j\in V(G)$ whose neighbors in $T$ are respectively $u_i$ and $u_j$, we have 
	\begin{equation}\label{eq:distance2}
		d_T(v_i,v_j)=|m(v_i) - m(v_j)| + w(u_iv_i) + w(u_jv_j).
	\end{equation}
	From Equations~\ref{eq:distance1} and \ref{eq:distance2}, we deduce that Equation~\ref{eq:distance} holds also for this direction for every $v_i\in B$ and $v_j\in V(G)$.
	By Fact~\ref{fact:intervals} and with symmetrical arguments to the ones of the previous direction, we conclude that $(B,R,(I_v)_{v\in V(G)})$ is a blue-red interval model of $G$.
\end{proof}


\begin{figure}[ht]
	\centering
	\includegraphics[width=0.95\linewidth]{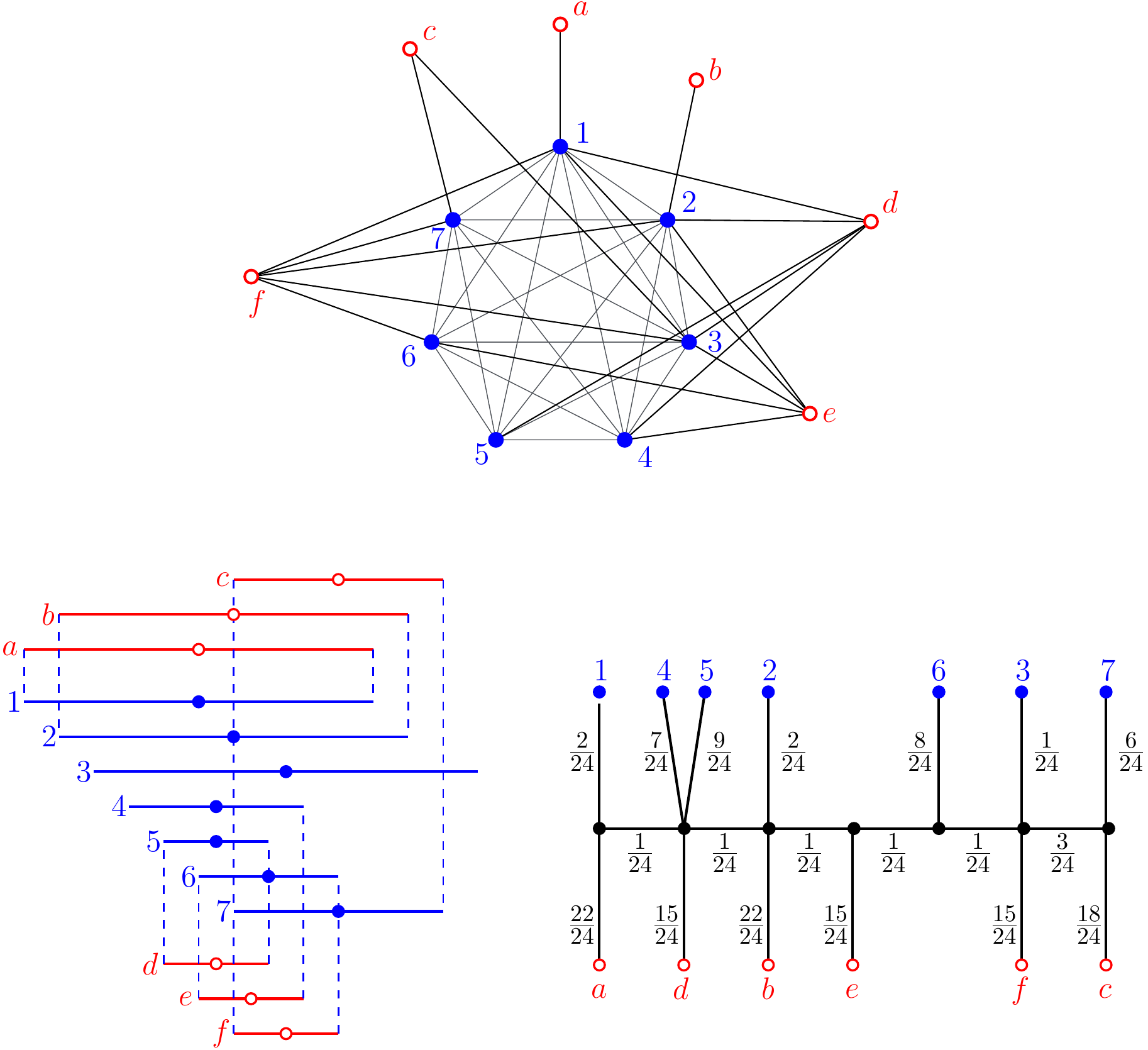}
	\caption{A split graph $G$ with a blue-red interval model of $G$ and a linear leaf root of $G$ obtained from the blue-red interval model by using the construction described in Theorem~\ref{thm:equivalence}. 
	Note the midpoint of each interval has been marked, to better understand where each vertex is placed in $T$. 
	Weights on edges is calculated as explained in the proof of Theorem~\ref{thm:equivalence} and note every edge has a weight that is a multiple of $\frac{1}{24}$ because the length of $I_3$ is 12.
	}
	\label{fig:fullexample}
\end{figure}

\section{Star NeS model}
\label{sec:starNES}

In this section, we first present an alternative definition of leaf powers through the notion of NeS models. 
We then show that we can recognize in polynomial time graphs with a \textit{star NeS model}: a NeS model whose embedding tree is a star (considering subdivided stars instead of stars does not make a difference).

For each tree $T$, we consider a corresponding tree $\cT$ embedded in the Euclidean plane so that each edge of $T$ corresponds to a line segment of $\cT$, these lines segments can intersect one another only at their endpoints, and the vertices of $T$ correspond (one-to-one) to the endpoints of the lines. 
Each line segment of $\cT$ has a positive Euclidean length. These embedded trees allow us to consider $\cT$ as the infinite set of points on the line segments of $\cT$. 
The notion of tree embedding used here is consistent with that found in \cite{Tamir1983}.
The line segments of $\cT$ and their endpoints are called respectively the edges and the nodes of $\cT$. The distance between two points $x,y$ of $\cT$ denoted by $d_{\cT}(x,y)$ is the length of the unique path in $\cT$  between $x$ and $y$ (the distance between two vertices of $T$ and their corresponding endpoints in $\cT$ are the same).

\begin{definition}[Neighborhood subtree, NeS-model] \label{def:NeS}
	Let $\cT$ be an embedding tree. For some point $c\in \cT$ and non-negative rational $w$, we define the neighborhood subtree with center $c$ and radius $w$ as the set of points $\{p \in \cT \mid  d_\cT(p,c) \leq w\}$.
	A NeS model $(\cT,(T_v)_{v\in V(G)})$ of a graph $G$ is a pair of an embedding tree $\cT$ and a collection of neigbhorhood subtrees of $\cT$ associated with each vertex of $G$ such that for every $u,v\in V(G)$, we have $uv\in E(G)$ iff $T_u \cap T_v\neq \emptyset$.
\end{definition}

\begin{theorem}	A graph is a leaf power iff it admits a NeS model.
\end{theorem}
\begin{proof}
	Brandstädt et~al. showed that leaf powers correspond to the graph class fixed tolerance NeST graph~\cite[Theorem 4]{BrandstadtHMW10}. Bidelnieks and Dearing showed that $G$ is a fixed tolerance NeST graph iff $G$ has a NeS model~\cite[Theorem 3.3]{BibelnieksD93}.
\end{proof}

\begin{figure}
	\centering
	\includegraphics[width=0.95\textwidth]{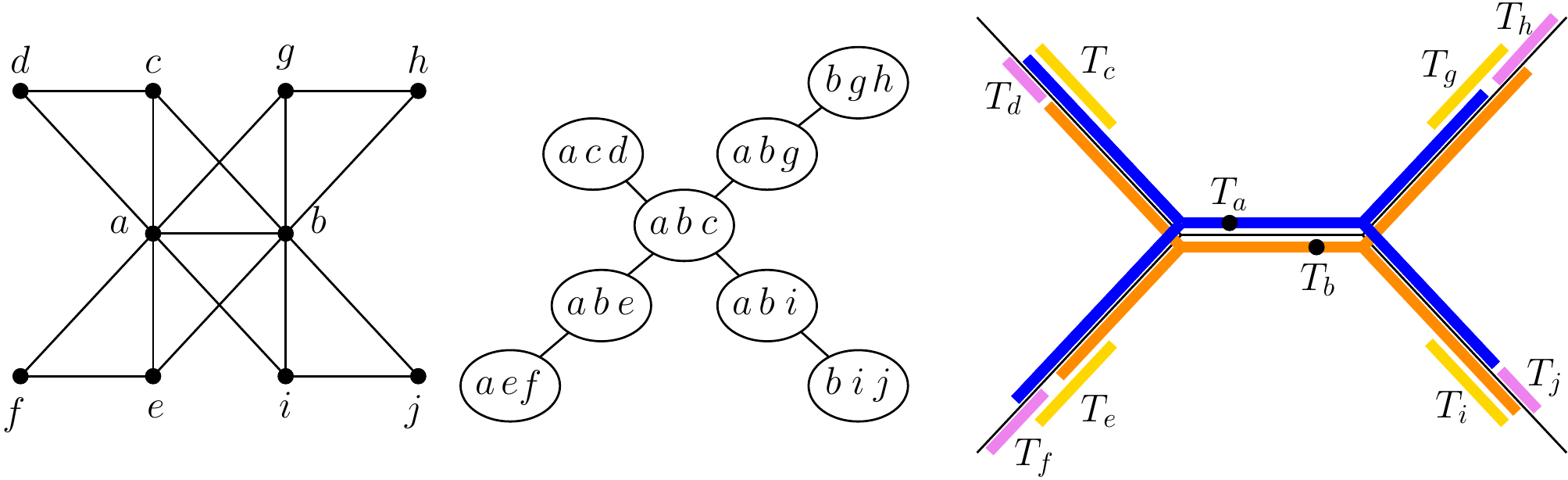}
	\caption{A graph with a clique tree and NeS model. The dots are centers of $T_a$ and $T_b$. 
	}
	\label{fig:CTandNES}
\end{figure}

See 
Figure \ref{fig:CTandNES} for a NeS model. 
Observe that every interval graph has a NeS model where the embedding tree is a single edge.
Moreover, if a graph $G$ admits a NeS model $(\cT,(T_v)_{v\in V(G)})$, then for every embedding path $\cL$ of $\cT$, $(\cL,(T_v\cap \cL)_{v\in X})$ is an interval model of $G[X]$ with $X$ the set of vertices $v$ such that $T_v$ intersects $\cL$.
As illustrated by the proofs of the following two lemmata, some results are easier to prove with NeS models than with other  characterizations.

We use the following operation on NeS models several times in the following proofs:
\begin{itemize}
    \item Given a NeS model $(\cT,(T_v)_{v\in V(G)})$ of a graph $G$ and $x\in V(G)$, if the center $c_x$ of $T_x$ is not an endpoint, we turn it into one by splitting the line containing $c_x$. We add a new line $L$ with endpoint $c_x$ that is sufficiently long so that $L$ contains a point $p$ that is not in any neighborhood subtree.  
    We replace the center of $T_x$ by $p$ and increase the radius of $T_v$ by $d_\cT(p,c_x)$.
\end{itemize}
We call this operation \emph{isolating} $T_x$ in $(\cT,(T_v)_{v\in V(G)})$. Observe that after this operation, $(\cT,(T_v)_{v\in V(G)})$ will remain a NeS model of $G$ since the intersection relations between $T_x$ and the other neighborhood subtrees do not change. Moreover, after this operation, $T_x$ is the only neighborhood subtree in $(T_v)_{v\in V(G)})$ that contains the center of $T_x$.

\begin{restatable}{lemma}{reductionrules}
\label{lem:reductionrules}
	For a graph $G$ and $u \in V(G)$ such that either (1)~$u$ is universal, or (2)~$u$ has degree 1, or (3)~$N(u)$ is a minimal separator in $G-u$, or (4)~$N(u)$ is a maximal clique in $G-u$.
	Then $G$ is a leaf power iff $G-u$ is a leaf power.
\end{restatable}
\begin{proof}
    The forward implication is trivial as being a leaf power is a hereditary property.
	For the backward implication, we assume we have a NeS model $(\cT,(T_v)_{v\in V(G-u)})$ of $G - u$, and in each case we show how to obtain a NeS model for $G$:
	\begin{itemize}
	    \item $u$ is universal, we simply create $T_u$ a neighborhood subtree with an arbitrary center and a radius sufficiently large so that $T_u$ intersects all the other neighborhood subtrees.
		\item $u$ has degree $1$ with $x$ as neighbor. We isolate $T_x$ in $(\cT,(T_v)_{v\in V(G-u)})$.  Then, we set $T_u=\{c_x\}$ where $c_x$ is the (new) center of $T_x$.
        \item $N(u)$ is a maximal clique in $G - u$. Then $T=\cap_{v\in N(u)}T_v$ is not empty and the only neighborhood subtrees intersecting $T$ are those associated with the vertices in $N(u)$.
        We set $T_u=\{c_u\}$ for some arbitrary point $c_u\in T$. 
        \item $N(u)$ is a minimal separator in $G - u$. Let $T=\cap_{v\in N(u)}T_v$. Since $T$ is the intersection of neighborhood subtrees, by \cite[Lemma 2.1]{BibelnieksD93}, $T$ is a neighborhood subtree and has a center $c_T$. We add a line $L$ to $\cT$ with endpoint $c_T$. 
        Let $c_u$ be the point on $L\cap T$ such that $d_\cT(c_T,c_u)$ is maximum.
        As $N(u)$ is a minimal separator, the only neighborhood subtrees containing $T$ are those associated with the vertices in $N(u)$. 
        By construction, we deduce that the only neighborhood subtrees containing $c_u$ are those associated with the vertices in $N(u)$. We set $T_u=\{c_u\}$.
	\end{itemize}
	It is not hard to see that in each case, $(\cT,(T_v)_{v\in V(G)})$ is a NeS model of $G$.
\end{proof}

\begin{lemma}\label{lem:cutvertex}
	Let $G$ be a graph and $u \in V(G$ a cut vertex. Then $G$ is a leaf power iff for every component $C$ of $G-u$, $G[V(C)\cup \{u\}]$ is a leaf power.
\end{lemma}
\begin{proof}
    As with Lemma \ref{lem:reductionrules}, the forward implication is trivial. For the backward implication, let $C_1,\dots,C_k$ be the components of $G- u$, and for every $1\leq i\leq k$ we let $C^+_i = G[V(C_i)\cup \{u\}]$. By assumption, every $C^+_i$ is a leaf power and there exists a NeS model $(\cT^i,(T^i_v)_{v\in V(C^+_i)})$ of $C^+_i$. 
    For each $1\leq i \leq k$, we do the following, we isolate $T_u^i$ in $(\cT_i,(T^i_v)_{v\in V(C^+_i)})$ and we make the new center $c_u^i$ of $T_u^i$ an endpoint by subdividing the new line containing $c_u^i$. Then, we multiply the length of every line of $\cT_i$ and the radius of every $(T^i_v)_{v\in V(C^+_i)}$ by $\frac{1}{r_u^i}$ where $r_u^i$ is the new radius of $T_u^i$.
    After these operations, for every $1\leq i \leq k$, the radius of $T_u^i$ is 1 and $T_u^i$ is the only neighborhood subtree in $(T_v)_{v\in V(C_i^+)})$ that contains the center of $T_u^i$.

    We make a new NeS model $(\cT,(T_v)_{v\in V(G)})$ where $\cT$ is obtained from the union of $\cT_1,\dots,\cT_k$ by identifying the centers (which are endpoints) of the neighborhood subtrees $T_u^1,\dots,T_u^k$ as a unique point $c_u$
    We define $T_u$ the neighborhood subtree with center $c_u$ and radius 1.
    Finally, for every $v\in V(G-u)$, we define $T_v=T_v^i$ with $i\in [k]$ such that $v\in V(C_i)$.
    It is straightforward to check that $(\cT,(T_v)_{v\in V(G)})$ is a NeS model of $G$.
\end{proof}




We now give the algorithm for recognizing graphs having a star NeS model. 
Our result is based on the purely combinatorial definition of \textit{good partition}, and we show that a graph admits a star NeS model iff it admits a good partition. Given a good partition, we compute a star NeS model in polynomial time. Finally, we prove that our Algorithm~\ref{algo:goodpartition} in polynomial time constructs a good partition of the input graph or confirms that it does not admit one.

Consider a star NeS model $(\cT,(T_v)_{v\in V(G)})$ of a graph $G$.
Observe that $\cT$ is the union of line segments $L_1,\dots,L_\beta$ with a common endpoint $c$ that is the center of $\cT$. Let $X$ be the set of vertices whose neighborhood subtrees contain $c$. For each $i\in[\beta]$, we let $B_i$ be the set of all vertices in $V(G)\setminus X$ whose neighborhood subtrees are subsets of $L_i$. The family $\cB = \{B_i \mid i\in[\beta]\}$ must then constitute a partition of $V(G)\setminus X$. We will show in Theorem~\ref{thm:goodpartition} that the pair $(X,\cB)$ has the properties of a good partition.

\begin{fact}\label{fact:structural}
Let $G$, $(\cT,(T_v)_{v\in V(G)})$ and $\cB$ be as defined above. We then have: 
\begin{itemize}
\item There is no edge between $B_i$ and $B_j$ for $i\neq j$ and thus $\cc(G-X) \sqsubseteq \cB$. 
\item For every $i \in[\beta]$ the NeS model $(L_i,(T_v\cap L_i)_{v\in B_i\cup X})$ is an interval model of $G[X \cup B_i]$.
\item For each $x\in X$ the neighborhood subtree $T_x$ 
is the union of the $\beta$ intervals $L_1\cap T_x,\dots,L_\beta\cap T_x$ and 
there exist positive rationals $\ell_x$ and $h_x$ with $\ell_x\leq h_x$ such that one interval among these intervals has length $h_x$ and the other $\beta -1$ intervals have length $\ell_x$.
\item If $\ell_x=h_x$, then the center of $T_x$ is $c$.
\end{itemize}
\end{fact}

\begin{claim}\label{claim:maximalclique}
	If $G$ has a star NeS model, it has a one, $(\cT,(T_v)_{v\in V(G)})$, where vertices whose neighborhood subtrees contain  the center of $\cT$ is a maximal clique.
\end{claim}
\begin{claimproof}
	Let $(\cT,(T_v)_{v\in V(G)})$ be the star NeS model of a graph $G$ with $c$ the center of $\cT$.
	Suppose that the set $X$ of vertices $x$ such that $c\in T_x$ is not a maximal clique.
	We deduce that there exists at least one maximal clique $Y$ containing $X$.
	For every such clique $Y$, $\cap_{y\in Y}T_y$ does not contain $c$ since $X\subset Y$.
	We take a maximal clique $Y$ such that $X\subset Y$ and $\min\{d_T(c,p)\mid p \in \cap_{y\in Y}T_y\}$ is minimum.
	Let $L$ be the line segment of $\cT$ containing $\cap_{y\in Y}T_y$.
	By Fact~\ref{fact:structural}, the neighborhood subtrees of the vertices in $Y\setminus X$ are intervals of $L$.
	We modify the NeS model by stretching the neighborhood subtrees of the vertices in $Y\setminus X$ so that they admit $c$ as an endpoint.
	After this operation, these subtrees remain intervals of $L$ and consequently, they still are neighborhood subtrees.
	Moreover, the choice of $Y$ implies that the only neighborhood subtrees that were intersecting the interval between $c$ and $\cap_{y\in Y}T_y$ are those associated with the vertices of $Y$.
	Thus, after this operation, we obtain a star NeS model where the set of vertices $x$ such that $c\in T_x$ is now the maximal clique $Y$.
\end{claimproof}


Claim \ref{claim:maximalclique} follows since we can always stretch some intervals to make $X$ a maximal clique. So far we have described a good partition as it arises from a star NeS model. Now we introduce the properties of a good partition that will allow to abstract away from geometrical aspects while still being equivalent, i.e. so that a graph has a good partition $(X, \cB)$ iff it has a star NeS model.
The first property is $\cc(G-X)\sqsubseteq \cB$ and the second is that for every $B \in \cB$ the graph $G[X \cup B]$ is an interval graph having a model where the intervals of $X$ contain the last point used in the interval representation.

\begin{definition}[$X$-interval graph] \label{def:Xinterval}
	Let $X$ be a maximal clique of $G$.
	We say that $G$ is an $X$-interval graph if $G$ admits a clique path $(K_1,\dots,K_k,X)$.
\end{definition}

The third property is the existence of an elimination order for the vertices of $X$ based on the lengths $\ell_x$ in the last item of Fact~\ref{fact:structural}, namely the permutation $(x_1,\dots,x_t)$ of $X$ such that $\ell_{x_1}\leq \ell_{x_2}\leq \dots \leq \ell_{x_t}$.
This permutation has the property that for any $i\in [t]$, among the vertices $x_i,x_{i+1},...,x_t$ the vertex $x_i$ must have the minimal neighborhood in at least $\beta -1$ of the blocks of $\cB$; we say that $x_i$ is \emph{removable} from $\{x_i,\dots,x_t\}$ for $\cB$.





\begin{definition}[Removable vertex] \label{def:removable}
	Let $X\subseteq V(G)$, $Y\subseteq X$ and let $\cB$ be a partition of $V(G)\setminus X$.
	Given a block $B$ of $\cB$ and $x\in Y$, we say $N(x)$ is minimal in $B$ for $Y$ if $N(x)\cap B \subseteq N(y)$ for every $y\in Y$.
	We say that a vertex $x\in Y$ is \textit{removable} from $Y$ for $\cB$ if $N(x)$ is minimal in at least $\abs{\cB}-1$ blocks of $\cB$ for $Y$.
\end{definition}



\begin{definition}[Good partition]
	\label{def:goodpartition}
	A good partition of a graph $G$ is a pair $(X,\cB)$ where $X$ is a maximal clique of $G$ and $\cB$ a partition of $V(G)\setminus X$ satisfying:
	\begin{enumerate}
		\item\label{item:sqsubset} $\cc(G-X)\sqsubseteq \cB$, i.e. every $C\in \cc(G-X)$ is contained in a block of $\cB$.
		\item\label{item:Xinterval} For each block $B\in \cB$, $G[X\cup B]$ is an $X$-interval graph.
		\item\label{item:elimination:order} There exists an elimination order $(x_1,\dots,x_t)$ on $X$ such that for every $i\in [t]$, $x_i$ is removable from $\{x_i,\dots,x_t\}$ for $\cB$.
	\end{enumerate}
	$X$ is the central clique of $(X,\cB)$ and $(x_1,\dots,x_t)$ a good permutation of  $(X,\cB)$.
\end{definition}


\begin{theorem}\label{thm:goodpartition}
	A graph $G$ admits a good partition iff it admits a star NeS model. Moreover, given the former we can compute the latter in polynomial time.
\end{theorem}
\begin{proof}
	($\Leftarrow$) Let $G$ be a graph with a star NeS model $(\cT,(T_v)_{v\in V(G)})$.
	Let $X$ be the set of vertices $x$ such that $T_x$ contains the center $c$ of the star $\cT$.
	By Claim~\ref{claim:maximalclique}, we can assume that $X$ is a maximal clique.
	As $\cT$ is a star, $\cT$ is an union of line segments $L_1,\dots,L_\beta$ with one common endpoint that is $c$.
	Let $\cB=\{B_1,\dots,B_\beta\}$ such that, for every $i\in [\beta]$, $B_i$ is the set of vertices $v\in V(G)\setminus X$ such that $T_v$ intersects $L_i$.
	We claim that $(X,\cB)$ is a good partition of $G$.
	Fact~\ref{fact:structural} implies that Property~\ref{item:sqsubset} and Property~\ref{item:Xinterval} are satisfied.
	For every $x\in X$, let $\ell_x$ be the rational defined in Fact~\ref{fact:structural}.
	Take a permutation $(x_1,\dots,x_t)$ of $X$ such that $\ell_{x_1}\leq \ell_{x_2}\leq \dots \leq \ell_{x_t}$.	Let $f(i)\in [\beta]$ such that the center of $x_i$ lies in $L_{f(i)}$. 
	From Fact~\ref{fact:structural}, we have $\abs{T_{x_i}\cap L_j}=\ell_{x_i}$ for every $j\neq f(i)$. Hence, for every $j\neq f(i)$ the interval $T_{x_i}\cap L_j$ are contained in the neighborhood subtree $T_{y}$ for every $y\in \{x_i,\dots,x_t\}$. 
	Consequently, $N(x_i)$ is minimal in $B_j$ for $\{x_i,\dots,x_t\}$ for every $j\neq f(i)$.
	We conclude that $(x_1,\dots,x_t)$ is a good permutation of $(X,\cB)$, i.e. Property~\ref{item:elimination:order} is satisfied.

    ($\Rightarrow$) Let $(X,\cB)$ be a good partition of a graph $G$ with $\cB=\{B_1,\dots,B_\beta\}$ and $(x_1,\dots,x_t)$ be a good permutation of $X$.
	For every $i\in [t]$, we define $X_{i}=\{x_i,\dots,x_t\}$. Since $x_i$ is removable from $X_i$, there exists an integer $f(i)\in[\beta]$ such that $N(x_i)$ is minimal for $X_i$ in every block of $\cB$ different from $B_{f(i)}$ (note that $f(i)$ is not necessarily unique, as $x_i$ could be minimal in every block of $\cB$ for $X_i$).
	
	Take $\cT$, the embedding of a star with center $c$ that is the union of $\beta$ line segments $L_1,\dots,L_\beta$ of length $2t+1$ whose intersection is $\{c\}$.
	We start by constructing the neighborhood subtree of the vertices in $X$.
	For doing so, we associate each $x_i\in X$ and each segment $L_j$ with a rational $\ell(x_i,L_j)$ and define $T_{x_i}$ as the union of over $j\in [\beta]$ of the points on $L_j$ at distance at most $\ell(x_i,L_j)$ from $c$ .
	
	For every $i\in [t]$ and $j\in [\beta]$ such that $j\neq f(i)$, we define $\ell(x_i,L_j)=i$.
	We define $\ell(x_t,L_{f(t)})=t$ and for every $i$ from $t-1$ to $1$ we define $\ell(x_i,L_{f(i)})$ as follows:
	\begin{itemize}
		\item If $N(x_i)$ is minimal in $B_{f(i)}$ for $X_i$ then we define $\ell(x_i,L_{f(i)})=i$.
		\item If $N(x_j)\cap B_{f(i)} \subset N(x_i)$ for every $x_j\in X_i$, then we define $\ell(x_i,L_{f(i)})=1+\max\{ \ell(y,L_{f(i)}) \mid y\in X_{i+1} \}$.
		\item Otherwise, we take $x_\min$ and $x_\max\in X_{i+1}$ such that 
		\[ N(x_\min)\cap B_{f(i)} \subseteq N(x_i)\cap B_{f(i)} \subseteq N(x_\max)\cap B_{f(i)} \]
		and $N(x_\min)\cap B_{f(i)}$ is maximal and $N(x_\max)\cap B_{f(i)}$ is minimal.
		We define $ \ell(x_i,L_{f(i)})=(\ell(x^\min,L_{f(i)})+\ell(x^\max,L_{f(i)}))/2. $
		Observe that the vertices $x_\min$ and $x_\max$ exist because $G[X\cup B_{f(i)}]$ is an $X$-interval graph and thus the neighborhoods of the vertices in $X$ in $B_{f(i)}$ are pairwise comparable for the inclusion.
	\end{itemize}
	
	By construction, we deduce the following properties on the lengths $\ell(x_i,L_j)$.
	
	\begin{claim}\label{claim:conditions}
		The following conditions hold for every $j\in [\beta]$:
		\begin{enumerate}
			\item For every $i\in [t]$, we have $\ell(x_i,L_j)\geq i$.
			\item For every $x,y\in X$, if $N(x)\cap B_j \subset N(y)$ then $\ell(x,L_j)< \ell(y,L_j)$.
		\end{enumerate}
	\end{claim}
	\begin{claimproof}
	We prove by induction on $i$ from $t$ to $1$ that Condition~(1) holds for $x_i$ and that Condition~(2) holds for every $x,y\in X_i$.
	That is obviously the case when $i=t$.
	Let $i\in[t-1]$ and suppose that Condition~(1) holds for every $x_{i+1},\dots,x_t$ and Condition~(2) holds for every $x,y\in X_{i+1}$.
	Let $j\in [\beta]$ such that $B_j\neq B_{f(i)}$.
	By construction, we have $\ell(x_i,L_j)=i$ so Condition~(1) holds for $x_i$ and $B_j$.
	Moreover, we have $N(x_i)\setminus B_{x_i} \subseteq N(y)$ for every $y\in X_i$.
	As $B_j\neq B_{f(i)}$, it follows that $N(x_i)\cap B_j \subseteq N(y)$ for every $y\in X_i$.
	The induction hypothesis implies that $\ell(y,L_j)\geq i+1$ for every $y\in X_i$.
	Hence, $\ell(x_i,L_j)<\ell(y,L_j)$ for every $y\in X_i$ and Condition~(2) holds for $X_i$ and every $B_j\neq B_{f(i)}$.
	
	It remains to prove that both Conditions holds for $B_{f(i)}$.
	We define $\ell(x_i,L_{f(i)})$ either to $i$, $a +1$ or $(b+c)/2$ where $a,b$ and $c$ belong to $\{\ell(y,L_{f(i)})\mid y\in X_{i+1}\}$.
	The induction assumption implies that $\ell(y,L_{f(i)})\geq i+1$ for every $y\in X_{i+1}$.
	Thus $a,b$ and $c$ greater than or equal to $i+1$.
	We deduce that $\ell(x_i,L_{f(i)})\geq i$. This proves that Condition~(1) holds for $x_i$.
	
	If $N(x_i)$ is minimal in $B_{f(i)}$ for $X_i$, then $\ell(x_i,L_j)<\ell(y,L_j)$ for every $y\in X_i$ and Condition~(2) is satisfied.
	Suppose that $N(x_i)$ is maximal in $B_{f(i)}$ for $X_i$, i.e. $N(y)\cap B_{f(i)}\subseteq N(x_i)$  for every $y\in X_{i}$.
	In this case, we set $\ell(x_i,L_{f(i)})$ is define as $1+ \max\{ \ell(y,L_{f(i)}) \mid y\in X_{i+1} \}$.
	Thus, we have $\ell(x_i,L_{f(i)})>\ell(y,L_{f(i)})$ for every $y\in X_{i+1}$ and Condition~(2) holds for $\leq_{B_{f(i)}}$.
	
	Finally, assume that $N(x_i)$ is neither minimal nor maximal in $B_{f(i)}$ for $X_i$.
	If $x^\max = x^\min$, then we have $\ell(x_i,L_{f(i)})=\ell(x^\max,L_{f(i)})$ and also $N(x_i)\cap B_{f(i)} = N(x^\min)\cap B_{f(i)}$.
	We deduce that Condition~(2) holds for every $x,y\in X_i$ because by induction hypothesis we have for every $x,y\in X_{i+1}$ and $x^\min \in X_{i+1}$.
	
	Now, assume that $x^\min\neq x^\max$.
	The way we choose $x^\min$ and $x^\max$ implies that 
	\[N(x^\min)\cap B_{f(i)} \subset N(x_i)\cap B_{f(i)} \subset N(x^\max) \cap B_{f(i)}.\]
	As $N(x^\min)\cap B_{f(i)} \subset N(x^\max)\cap B_{f(i)}$, by induction hypothesis, we deduce that $\ell(x^\min,L_{f(i)}) < \ell(x^\max,L_{f(i)})$ and thus we have
	\[ \ell(x^\min,L_{f(i)}) < \ell(x_i,L_{f(i)}) < \ell(x^\max,L_{f(i)})\]
	So Condition~(2) holds for every $x,y\in \{x_i,x^\min, x^\max\}$.
	
	Let $y\in X_{i+1}$ such that $N(x_i)\cap B_{f(i)} \subseteq N(y)$.
	The minimality of $N(x^\max)\cap B_{f(i)}$ implies that $N(x^\max)\cap B_{f(i)} \subseteq N(y)$.
	By induction hypothesis, we have $\ell(x^\max,L_{f(i)})\leq \ell(y,L_{f(i)})$.
	Thus, for every $y\in X_{i+1}$ such that $N(x_i)\cap B_{f(i)} \subseteq N(y)$, we have $\ell(x_i,L_{f(i)}) < \ell(y,L_{f(i)})$.
	
	Symmetrically, the maximality of $x^\min$ implies that for every $y\in X_{i+1}$ such that $N(y)\cap B_{f(i)} \subseteq N(x_i)$, we have $\ell(y,L_{f(i)}) < \ell(x_i,L_{f(i)})$.
	We conclude that Condition~(2) holds for every $x,y\in X_i$.
	By induction, we conclude that Condition~(1) holds for every $i\in[t]$ and Condition~(2) holds for every $x,y\in X$.
	This concludes the proof of Claim~\ref{claim:conditions}. 
\end{claimproof}
	
	Observe that each $T_{x_i}$ is a neighborhood subtree as by construction and Condition~(1) of Claim~\ref{claim:conditions} the lengths $\ell(x_i,L_j)$ satisfy the last item of Fact~\ref{fact:structural}. See Figure~\ref{fig:star} for an illustration of this neighborhood subtree.
	
	\begin{figure}
		\centering
		\includegraphics[width=0.7\linewidth]{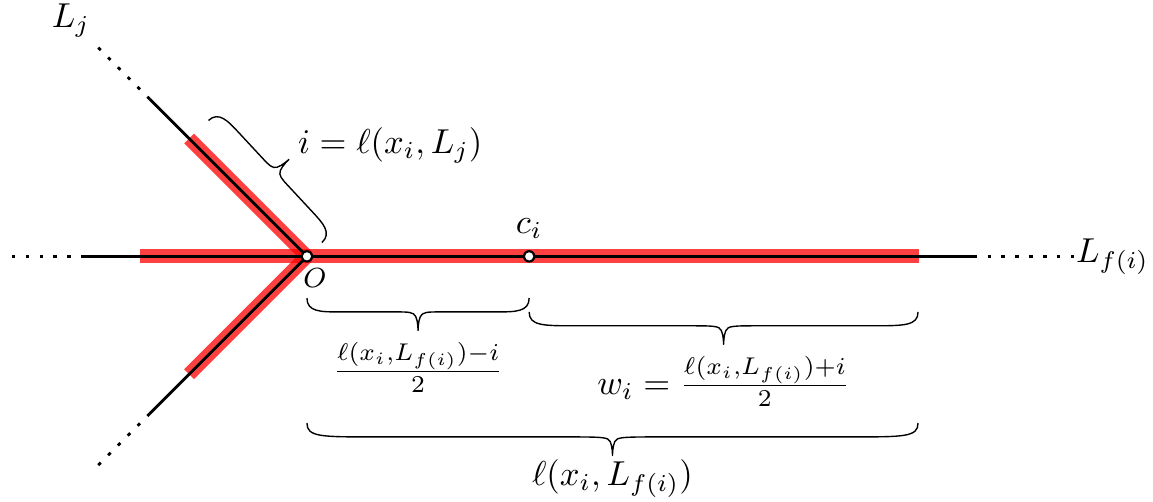}
		\caption{Illustration of the neighborhood subtree $T_{x_i}$ where the number $\beta$ of edges of the star NeS model is 4.}
		\label{fig:star}
	\end{figure}
	
	It remains to construct the neighborhood subtrees of the vertices in $V(G)\setminus X$.
	Let us explain how we do it for the vertices in $B_j$ for some $j\in[\beta]$. 
	For every $b\in B_j$, we define $T_b$ as an interval on the line segment $L_j$.
	As $(X,\cB)$ is a good partition, $G[X\cup B_j]$ admits a clique path $(X,K_1,\dots,K_k)$.
	Clique path properties implies that
	$X\supseteq K_1\cap X\supseteq \dots \supseteq K_k \cap X$.
	For every $i\in[k]$, $x\in X\setminus K_i$ and $y\in X\cap K_i$, we have $N(x)\cap B_j \subset N(y)$ and thus $\ell(x,L_j)<\ell(y,L_j)$ thanks to Condition~(2) of Claim~\ref{claim:conditions}.
	As $\ell(x,L_j)$ corresponds to the length of the interval $T_x\cap L_j$ for every $x\in X$, we conclude that there exist $p_1,\dots,p_k$ points on $L_j$ such that
	\begin{itemize}
		\item for every $i\in [k]$, $K_i\cap X$ is exactly the set of vertices $x$ in $X$ such that $p_i\in T_x$ and
		\item we have $0<d_\cT(p_1,c)<d_\cT(p_{2},c)< \dots <d_\cT(p_{k},c)$.
	\end{itemize}
	For every $b\in B_j$ such that $b$ is contained in every clique $K_i$ for $i$ between $s$ and $t$, we define $T_b$ as the interval of $L_j$ between $p_s$ and $p_t$.
	
	By construction, $(\cT,(T_v)_{v\in V(G)})$ fulfill every property of Fact~\ref{fact:structural}.
	We deduce that $(\cT,(T_v)_{v\in V(G)})$ is a NeS model of $G$.
	Obviously, the construction of  $(\cT,(T_v)_{v\in V(G)})$ can be done in polynomial time.
\end{proof}
	
It is easy to see that every graph that admits a star NeS model has a clique-tree that is a subdivided star. The converse is not true. In fact, for every graph $G$ with a clique tree that is a subdivided star with center $X$, the pair $(X,\cc(G-X))$ satisfies Properties~\ref{item:sqsubset} and~\ref{item:Xinterval} of Definition~\ref{def:goodpartition} but Property~\ref{item:elimination:order} might not be satisfied. See for example the graph in Figure~\ref{fig:CTandNES} and note that the pair $(\{a,b,c\},\cc(G-\{a,b,c\}))$ does not satisfy Property~\ref{item:elimination:order}, as after removing the vertex $c$ neither $a$ nor $b$ is removable from $\{a,b\}$. 

We now describe Algorithm~\ref{algo:goodpartition} that decides whether a graph $G$ admits a good partition. Clearly $G$ must be chordal, so we start by checking this. A chordal graph has $O(n)$ maximal cliques, and  for each maximal clique $X$ we try to construct a good partition $(X,\cA)$ of $G$.
We start with $\cA \leftarrow \cc(G-X)$ and note that $(X,\cA)$ trivially satisfies Property~\ref{item:sqsubset} of Definition~\ref{def:goodpartition}. Moreover, if $G$ admits a good partition with central clique $X$, then  $(X,\cA)$ must satisfy Property~\ref{item:Xinterval} of Definition~\ref{def:goodpartition}, and we check this in Line~\ref{line:checking1}.
Then, Algorithm~\ref{algo:goodpartition} iteratively in a while loop tries to construct a good permutation $(w_1,\dots,w_{\abs{X}})$ of $X$, while possibly merging some blocks of $\cA$ along the way so that it satisfies Property~\ref{item:elimination:order}, or discover that there is no good partition with central clique $X$.

For doing so, at every iteration of the while loop, Algorithm~\ref{algo:goodpartition} searches for a vertex $w$ in $W$ (the set of unprocessed vertices) such that $\notmin(w,W,\cA)$ -- the union of the blocks of $\cA$ where $N(w)$ is not minimal for $W$ (see Definition~\ref{def:notmin}) -- induces an $X$-interval graph with $X$.
If such a vertex $w$ exists, then Algorithm~\ref{algo:goodpartition} sets $w_r$ to $w$, increments $r$ and merges the blocks of $\cA$ contained in $\notmin(w_r,W,\cA)$ (Line~\ref{line:constructAp}) to make $w_r$ removable in $\cA$ for $W$. 
Otherwise, when no such vertex $w$ exists, Algorithm~\ref{algo:goodpartition} stops the while loop (Line~\ref{line:checking}) and tries another candidate for $X$.
For the graph in Figure~\ref{fig:CTandNES}, with $X=\{a,b,c\}$ the first iteration of the while loop will succeed and set $w_1=c$, but in the second iteration neither $a$ nor $b$ satisfy the condition of Line~\ref{line:choosingw}.

At the start of an iteration of the while loop, the algorithm has already choosen the vertices $w_1,w_2,...,w_{r-1}$ and $W=X\setminus \{w_1,\dots,w_{r-1}\}$. For every $i\in[r-1]$, $w_i$ is removable from $X\setminus\{w_1,\dots,w_{i-1}\}$ for $\cA$.
According to Definition \ref{def:removable} the next vertex $w_r$ to be removed should have $N(w_r)$ non-minimal for $W$ in at most one block of the good partition we want to construct. However, the neighborhood $N(w_r)$ may be non-minimal for $W$ in several blocks of the current partition $\cA$, since these blocks may be (unions of) separate components of $G \setminus X$ that should live on the same line segment of a star NeS model and thus actually be a single block which together with $X$ induces an $X$-interval graph. An example of this merging is given in Figure~\ref{fig:starNES}. The following definition captures, for each $w\in W$, the union of the blocks of $\cA$ where $N(w)$ is not minimal for $W$.


\begin{definition}[$\notmin$]\label{def:notmin}
	For $W\subseteq X$,  $x\in W$ and partition $\cA$ of $V(G)\setminus X$, we denote by $\notmin(x,W,\cA)$ the union of the blocks $A\in \cA$ where $N(x)$ is not minimal in $A$ for $W$.
\end{definition}

\begin{algorithm}[bth]
	\SetAlgoLined
	\KwIn{A graph $G$.}
	\KwOut{A good partition of $G$ or ``\textsf{no}''.}
	Check if $G$ is chordal and if so, compute its set of maximal cliques $\cX$, otherwise \Return \textsf{no}\;
	\For{every $X\in \cX$}
	{
		\lIf{there exists $C\in \cc(G-X)$ such that $G[X\cup C]$ is not an $X$-interval graph}{\KwSty{continue}}\label{line:checking1}
		$\cA\leftarrow \cc(G-X)$, $W\leftarrow X$ and $r\leftarrow 1$\;
		\While{$W\neq \emptyset$}
		{
		    \lIf{there exists $w\in W$ such that $G[X\cup \notmin(w_r,W,\cA)]$ is an $X$-interval graph\label{line:choosingw}}{$w_r\leftarrow w$}
			\lElse{\KwSty{break}}\label{line:checking}
			Replace the blocks of $\cA$ contained in $\notmin(w_r,W,\cA)$ by $\notmin(w_r,W,\cA)$\label{line:constructAp}\;
			$W\leftarrow W\setminus \{w_r\}$ and $r\leftarrow r+1$\;
		}
		\lIf{$W=\emptyset$}{\Return $(X,\cA)$}
	}
	\Return \textsf{no}\;
	\caption{}
	\label{algo:goodpartition}
\end{algorithm}

As we already argued, when Algorithm~\ref{algo:goodpartition} starts a while loop, $(X,\cA)$ satisfies Properties~\ref{item:sqsubset} and~\ref{item:Xinterval}, and it is not hard to argue that in each iteration, for every $i\in[r-1]$, $w_i$ is removable from $X\setminus\{w_1,\dots,w_{i-1}\}$ for $\cA$. Hence, when $W=\emptyset$, then $(w_1,\dots,w_{\abs{X}})$ is a good permutation of $(X,\cA)$ and Property~\ref{item:elimination:order} is satisfied.

\begin{lemma}\label{lem:goodOutput}
	If Algorithm~\ref{algo:goodpartition} returns $(X,\cB)$, then $(X,\cB)$ is a good partition.
\end{lemma}
\begin{proof}
	Assume that Algorithm~\ref{algo:goodpartition} returns a pair $(X,\cB)$.
	We show that  $(X,\cB)$ satisfies Definition \ref{def:goodpartition}.
    As the algorithm starts with $\cA=\cc(G-X)$ and only merge blocks, Property~\ref{item:sqsubset} holds.
    Line~\ref{line:checking} guarantees that every block of $\cc(G-X)$ induces with $X$ an $X$-interval graph.
    Moreover, when merging blocks, Line \ref{line:choosingw} checks that the merged block induces with $X$ an $X$-interval graph. Thus Property~\ref{item:Xinterval} holds.
	
    Observe that at every step of the while loop, we have $\cA\sqsubseteq \cB$ since we only merge blocks of $\cA$ and $\cB$ is the last value of $\cA$.
	We show that $(w_1,\dots,w_t)$ is a good permutation of $(X,\cB)$ where, for each $i\in[t]$, $w_i$ is the vertex chosen by the algorithm at the $i$-th iteration.
	Let $i\in [t]$. At the start of the $i$-th iteration of the while loop, the value of $W$ is $X_i=\{w_i,\dots,w_t\}$.
	During this iteration, we merge the blocks of $\cA$ contained in $\notmin(w_i,X_i,\cA)$.
	Consequently, after the merging, there exists a block $A_{w_i}$ of $\cA$ such that $N(w_i)$ is minimum for $X_i$ in every block of $\cA$ different from $A_{w_i}$.
	As $\cA\sqsubseteq \cB$, there exists a block $B_{w_i}$ of $\cB$ containing $A_{w_i}$.
	We deduce that, $N(w_i)$ is minimal for $X_i$ in every block of $\cB$ different from $B_{w_i}$, i.e. $w_i$ is removable from $X_i$ for $\cB$.
	Thus, $(w_1,\dots,w_t)$ is a good permutation of $(X,\cB)$ and we conclude that $(X,\cB)$ is a good partition of $G$.
\end{proof}

To prove the opposite direction, namely that if $G$ has a good partition  $(X,\cB)$ associated with a good permutation $(x_1,\dots,x_t)$, then Algorithm \ref{algo:goodpartition} finds a good partition, we need two lemmata.
 The easy case is when Algorithm~\ref{algo:goodpartition} chooses consecutively $w_1=x_1,\dots,w_t=x_t$, and we can use Lemma~\ref{lem:goodPermut} to prove that it will not return \textsf{no}.
 However, Algorithm~\ref{algo:goodpartition} does not have this  permutation as input and at some iteration with $w_1=x_1,\dots,w_{r-1}=x_{r-1}$, the algorithm might stop to follow the permutation $(x_1,\dots,x_t)$ and choose a vertex $w_r=x_i$ with $r<i$ because $x_r$ may not be the only vertex satisfying the condition of Line~\ref{line:choosingw}.
 In Lemma \ref{lem:noMistake} we show that choosing $w_r=x_i$ is then not a mistake as it implies the existence of another good partition and another good permutation that starts with $(x_1,\dots,x_{r-1},w_r=x_i)$. See Figure~\ref{fig:starNES} for an example of a very simple graph with several good permutations leading to quite distinct star NeS models.

\begin{figure}
	\centering
	\includegraphics[width=\textwidth]{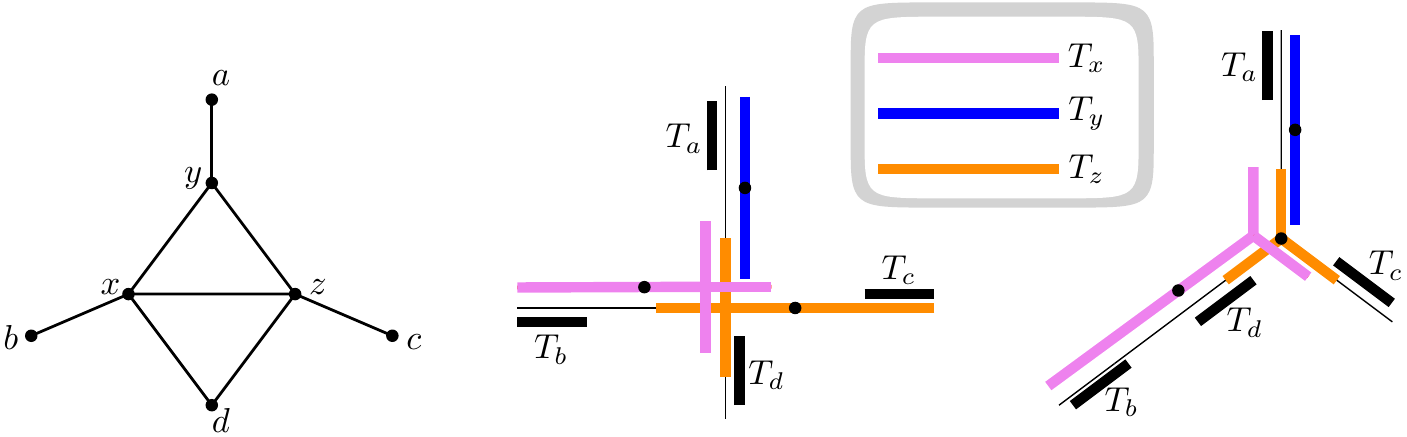}
	\caption{A graph $G$ and two star NeS models. 
	The dots are the centers of the subtrees.
	$X = \{x,y,z\}$ is a maximal clique with four components in $G-X$. The left NeS model corresponds to permutations $(y,x,z)$ or $(y,z,x)$ and the right to permutation $(x,y,z)$.
	If Algorithm~\ref{algo:goodpartition} chooses $w_1=x$, then the components $\{b\}$ and $\{d\}$ will be merged in Line~\ref{line:constructAp}.
	There is a third star NeS model, corresponding to permutation $(z,y,x)$, similar to the one on the right. The last two permutations of $X$ are not good permutations.}
	\label{fig:starNES}
\end{figure}

We need some definitions. 
Given permutation $P=(x_1,\dots,x_\ell)$ of a subset of $X$ and $i\in[\ell]$, define $\cA_0^P=\cc(G-X)$ and $\cA_i^P$ the partition of $V(G)\setminus X$ obtained from $\cA_{i-1}^P$ by merging the blocks contained in $\notmin(x_i,X\setminus \{x_1,\dots,x_{i-1}\},\cA_{i-1}^P)$. 
Observe that  when Algorithm~\ref{algo:goodpartition} treats $X$ and we have $w_1=x_1,\dots,w_\ell=x_\ell$, then the values of $\cA$ are successively $\cA_0^P,\dots,\cA_\ell^P$.
The following lemma proves that if there exists a good permutation $P=(x_1,\dots,x_t)$ and at some iteration we have $w_1=x_1,\dots,w_{r-1}=x_{r-1}$, then the vertex $x_r$ satisfies the condition of Line~\ref{line:choosingw} and Algorithm~\ref{algo:goodpartition} does not return \textsf{no} during this iteration.
Thus, as long as Algorithm~\ref{algo:goodpartition} follows a good permutation, it will not return \textsf{no}.

\begin{lemma}\label{lem:goodPermut}
	Let $G$ be a graph with good partition $(X,\cB)$ and $P=(x_1,\dots,x_t)$ be a good permutation of $(X,\cB)$.
	For every $i\in [t]$, we have $\cA^P_{i}\sqsubseteq \cB$ and the graph $G[X\cup \notmin(x_{i},\{x_{i},\dots,x_t\},\cA^P_{i-1})]$ is an $X$-interval graph.
\end{lemma}
\begin{proof}
	For every $i\in [t]$, we denote $\{x_i,\dots,x_t\}$ by $X_i$.
	We start by proving by induction that, for every $i\in[t]$, we have $\cA^P_i\sqsubseteq \cB$. This is true for $\cA^P_0=\cc(G-X)$ by Property~\ref{item:sqsubset} of Definition~\ref{def:goodpartition}.
	Let $i\in[t]$ and suppose that $\cA^P_{i-1}\sqsubseteq \cB$.
	Every block of $\cA^P_i$ different from $\notmin(x_i,X_i,\cA^P_{i-1})$ is a block of $\cA^P_{i-1}$ and is included in a block of $\cB$ by the induction hypothesis.
	As $x_i$ is removable from $X_i$ for $\cB$, there exists $B_{x_i}\in\cB$ such that $N(x_i)$ is minimal for $X_i$ in every block of $\cB$ different from $ B_{x_i}$.
	Since $\cA^P_{i-1}\sqsubseteq \cB$, for every block $\widehat{A}\in \cA^P_{i-1}$ such that $N(x_i)$ is not minimal in $\widehat{A}$ for $X_i$, we have $\widehat{A}\subseteq  B_{x_i}$.
    As the union of these blocks $\widehat{A}$'s is $\notmin(x_i,X_i,\cA^P_{i-1})$, we deduce that $\notmin(x_i,X_i,\cA^P_{i-1})\subseteq  B_{x_i}$.
	Hence, every block of $\cA^P_i$ is included in a block of $\cB$, thus by induction $\cA^P_i\sqsubseteq \cB$ for every $i\in [t]$.
	
	For every $i\in[t]$, we have proved that $\notmin(x_{i},X_i,\cA^P_{i-1})\subseteq B_{x_i}$.
	From Property~\ref{item:Xinterval} of Definition~\ref{def:goodpartition}, $G[X\cup  B_{x_i}]$ is an $X$-interval graph. We conclude that $G[X\cup \notmin(x_{i},X_i,\cA^P_{i-1})]$ is also an $X$-interval graph.
\end{proof}

\begin{lemma}\label{lem:noMistake}
	Let $P=(x_1,\dots,x_t)$ be a good permutation of $X$ and $i\in [t]$.
	For every $w\in \{x_i,\dots,x_t\}$ such that $G[X\cup \notmin(w,\{x_i,\dots,x_t\},\cA_{i-1}^P)]$ is an $X$-interval graph, there exists a good permutation of $X$ starting with $(x_1,\dots,x_{i-1},w)$.
\end{lemma}
\begin{proof}
	This lemma relies on the following relation that will be applied to vertex subsets corresponding to unions of connected components of $G-X$.
	
	\begin{definition}\label{def:leqX}
		For every $A\subseteq V(G)\setminus X$, we define $\maxcap(A)=N(A)\cap X$ and $\mincap(A)=\cap_{a\in A}N(a)\cap X$.
		Given $A_1,A_2\subseteq V(G)\setminus X$,  we say that $A_1 \leq_X A_2$ if $\maxcap(A_1)\subseteq \mincap(A_2)$.
	\end{definition}
	
		As an example of the use of this, note that in Figure \ref{fig:starNES}, we have $\{b\} \leq_{\{x,y,z\}} \{d\}$ and in the rightmost star NeS model we see this ordering reflected by $d$ being placed closer to the center of the star than $b$.

	For every $A\subseteq V(G)\setminus X$, we have $\mincap(A)\subseteq \maxcap(A)$ and thus $\leq_X$ is transitive.
	The following claim reveals the connection between Definition~\ref{def:leqX} and Property~\ref{item:Xinterval} of Definition~\ref{def:goodpartition}.
	
	\begin{claim}\label{claim:comparable}
		Let $P=(x_1,\dots,x_t)$ be a good permutation of $X$ and $\ell\in \{0,\dots,t\}$.
		For all distinct blocks $A,A^\star\in \cA^P_\ell$, if $G[X\cup A\cup A^\star]$ is an $X$-interval graph, then $A$ and $A^\star$ are comparable for $\leq_X$.
    \end{claim}
	\begin{claimproof}
	  	We prove this claim by induction on $\ell$.
		By definition, we have $\cA_0^P=\cc(G-X)$. Let $C,C^\star$ be two distinct components of $G-X$ such that $G[X\cup C \cup C^\star]$ is an $X$-interval graph.
		Let $(K_1,\dots,K_k,X)$ be a clique path of $G[X\cup C \cup C^\star]$.
		Since $C$ and $C^\star$ are distinct connected components, for every $j\in [k]$, $K_j$ is either a maximal clique of $G[X\cup C]$ or $G[X\cup C^\star]$.
		Assume w.l.o.g. that $K_1$ is a maximal clique of $G[X\cup C]$.
		Moreover, there exists $j\in [k-1]$ such that $X,K_1,K_2,\dots,K_j$ are all the maximal cliques of $G[X\cup C]$ and $X,K_{j+1},\dots,K_k$ are all the maximal cliques of $G[X\cup C^\star]$.
		Since $(K_1,\dots,K_k,X)$ is a clique path, we have $K_1\cap X\subseteq \dots \subseteq K_k\cap X\subseteq X$.
	    From Definition~\ref{def:leqX}, we deduce that $\maxcap(C)=K_j\cap X$ and $\mincap(C^\star)=K_{j+1}\cap X$.
		As $K_j\cap X \subseteq K_{j+1}\cap X$, we have $\maxcap(C)\subseteq \mincap(C^\star)$ and thus $C\leq_X C^\star$ and the claim holds for $\cA_0^P$.
		
		Let $\ell\in [t]$ and suppose that the claim holds for $\cA_{\ell-1}^P$.
		In the rest of this proof, we use the shorthand $\notmin(x_\ell)$ to denote $\notmin(x_\ell,\{x_\ell,\dots,x_t\},\cA_{\ell-1}^P)$.
		We obtain $\cA_{\ell}^P$ from $\cA_{\ell-1}^P$ by merging the blocks of $\cA_{\ell-1}^P$ contained in $\notmin(x_\ell)$.
		So every block of $\cA_{\ell}^P$ different from $\notmin(x_\ell)$ is also a block of $\cA_{\ell-1}^P$.
		By assumption, the claim holds for every pair of distinct blocks of $\cA_{\ell}^P \setminus \{\notmin(x_\ell)\}$.
		To prove that it holds for $\cA_\ell^P$, it is enough to prove that it holds for every pair $(A,\notmin(x_\ell))$ with $A\neq \notmin(x_\ell)$.
		
		Let $A\in \cA_{\ell}^P$ such that $A\neq \notmin(x_\ell)$ and $G[X\cup A\cup \notmin(x_\ell)]$ is an $X$-interval graph.
		We need to prove that $A\leq_X \notmin(x_\ell)$ or $\notmin(x_\ell)\leq_X A$.
		Observe that $\notmin(x_\ell)$ is a union of blocks $A_1,\dots,A_k$ of $\cA_{\ell-1}^P$.
		As $G[X\cup A \cup \notmin(x_\ell)]$ is an $X$-interval graph, also for any two blocks $B,C$ out of $A,A_1,\dots,A_k$ the graph $G[X\cup B\cup C]$ is an $X$-interval graph, and thus by induction hypothesis, we deduce that every pair of blocks among $A,A_1,\dots,A_k$ is comparable for $\leq_X$.
		Suppose w.l.o.g. that $A_1\leq_X \dots\leq_X A_k$.
		
		Assume towards a contradiction that there exists $j\in [k-1]$ such that $A_j\leq_X A \leq_X A_{j+1}$.
		By Definition~\ref{def:notmin}, we have 
		\[ \maxcap(A_j) \subseteq \mincap(A)\subseteq \maxcap(A)\subseteq \mincap(A_{j+1}). \]
		Since $N(w)$ is not minimal in $A_j$ for $X_\ell$, we have $N(w)\cap A_j\neq \emptyset$ and thus $w\in \maxcap(A_j)$.
		We deduce that $w$ belongs to $\mincap(A)$ and $\mincap(A_{j+1})$ and thus $A,A_{j+1}\subseteq N(w)$.
		As $N(w)$ is not minimal in $A_{j+1}$ and $A_{j+1}\subseteq N(w)$, we deduce that there exists $y\in X_\ell$ such that $A_{j+1}\not\subseteq N(y)$ which means that $y\notin \mincap(A_{j+1})$.
		Since $\maxcap(A)\subseteq \mincap(A_{j+1})$, we have $y\notin \maxcap(A)$.
		Hence, $A\subseteq N(w)$ and $N(y)\cap A=\emptyset$.
		We conclude that $N(w)$ is not minimal in $A$ for $X_\ell$, a contradiction with $A\neq\notmin(x_\ell)$ because $N(w)$ is minimal for $X_\ell$ in every block of $\cA_\ell^P$ different from $\notmin(x_\ell)$.
		
		It follows that $A\leq_X A_1$ or $A_k\leq_X A$.
		Since $\notmin(x_\ell)=A_1\cup\dots\cup A_k$, we deduce that $A\leq_X \notmin(x_\ell)$ or $\notmin(x_\ell)\leq_X A$, that is the claim holds for $\cA_{\ell}^P$.
		By induction, we conclude that the claim is true for every $\ell\in \{0,\dots,t\}$.
	\end{claimproof}
	
	\begin{claim}\label{claim:leqX}
		Let $X\subseteq V(G)$ and $\cA$ a partition of $V(G)\setminus X$ such that $\cc(G-X)\sqsubseteq \cA$.
		For any $A_1,\dots,A_k\in \cA$ that are pairwise comparable for $\leq_X$, if for each $j\in [k]$ $G[X\cup A_j]$ is an $X$-interval graph, then $G[X\cup A_1\cup \dots \cup A_k]$ is an $X$-interval graph.
    \end{claim}
	\begin{claimproof}
	    Let $A_1,\dots,A_k$ be $k$ blocks of $\cA$ such that $A_1\leq_X A_2\leq_X \dots \leq_X A_k$.
	    Assume that $G[X\cup A_j]$ is an $X$-interval graph for each $j\in [k]$.
	    
	    Since $G[X\cup A_j]$ is an $X$-interval graph for every $j\in[k]$, $G[X\cup A_j]$ admits a clique path $(K_1^j,\dots,K_{f(i)}^j,X)$, where $f(j)$ is the number of maximal cliques minus one in $G[X\cup A_j]$.
	    For every $j\in [k]$, we have $K_1^j \cap X \subseteq K_2^j\cap X \subseteq \dots \subseteq K_{f(j)}^j\cap X$.
	    Moreover, by definition, we have $\mincap(A_i)=K_1^j$ and $\maxcap(A_i)=K_{f(j)}^j$.
	    For every $j\in [k-1]$, since $A_j\leq_X A_{j+1}$, we have $\maxcap(A_j)\subseteq \mincap(A_{j+1})$, that is $K_{f(j)}^j \cap X \subseteq K_1^{j+1}\cap X$.
	    Because $\cc(G-X) \sqsubseteq \cA$, every maximal clique of $G[X\cup A_1\cup \dots \cup A_k]$  different from $X$ is a maximal clique of a unique $G[X\cup A_j]$.
	    We conclude that $(K_1^1,\dots,K_{f(1)}^1,\dots,K_1^k,\dots,K_{f(k)}^k,X)$ is a clique path of $G[X\cup A_1\cup \dots \cup A_k]$.
	\end{claimproof}
        
		For every $\ell\in [t]$, we denote by $X_\ell$ the set $\{x_\ell,\dots,x_t\}$.
		Let $w\in X_i$ such that $G[X\cup \notmin(w,X_i,\cA_{i-1}^P)]$ is an $X$-interval graph.
		Let $(X,\cB)$ be a good partition of $G$ such that $(x_1,\dots,x_t)$ is a good permutation of $(X,\cB)$. 
		If $w=x_p$ is removable from $\cB$ for $X_i$, then $(x_1,\dots,x_{i-1},w,x_i,\dots,x_{p-1},x_{p+1},\dots,x_t)$ is a good permutation of $(X,\cB)$ and we are done.
		In particular, $w$ is removable from $\cB$ if $\abs{\notmin(w,X_i,\cA_{i-1}^P)}\leq 1$.
		In the following, we assume that $\abs{\notmin(w,X_i,\cA_{i-1}^P)}\geq 2$.
		
		We construct a good partition $(X,\cB_\new)$ that admits good permutation starting with $(x_1,\dots,x_{i-1},w)$.
		Let $A_1,\dots,A_k\in \cA_{i-1}^P$ such that $\notmin(w,X_i,\cA_{i-1}^P)=A_1\cup A_2 \cup \dots\cup A_k$.
		As $G[X\cup \notmin(w,X_i,\cA_{i-1}^P)]$ is an $X$-interval graph, by Claim~\ref{claim:comparable}, the blocks $A_1,\dots,A_k$ are pairwise comparable for $\leq_X$.
		Suppose w.l.o.g. that $A_1\leq_X \dots \leq_X A_k$. In Figure~\ref{fig:starNES} with $\notmin(x,\{x,y,z\},\{\{a\},\{b\},\{c\},\{d\}\})=\{b\} \cup \{d\}$ we have $A_1=\{b\}$ and $A_2=\{d\}$.
		
		By Lemma~\ref{lem:goodPermut}, we have $\cA_{i-1}^P \sqsubseteq \cB$.
		Thus, there exists a block $B_1$ of $\cB$ containing $A_1$ and $B_1$ is a union of blocks of $\cA_{i-1}^P$.
 		Let $B_1^\max$ be the union of all the blocks $A$ of $\cA_{i-1}^P$ included in $B_1$ such that $A$ is not contained in $\notmin(w,X_i,\cA_{i-1}^P)$ and $A_1\leq_X A$.
		Note that for every $A\in B_1^\max$, we have $A_k\leq_X A$ because otherwise we have $A_1\leq_X A \leq_X A_k$ and that implies $A\subseteq \notmin(w,X_i,\cA_{i-1}^P)$ (see the arguments used in the proof of Claim~\ref{claim:comparable}).
		
		\begin{claim}\label{claim:Bmax}
		    We can assume that $B_1^\max$ is empty.
		\end{claim}
		\begin{claimproof}
			Suppose that $B_1^\max\neq \emptyset$.
			Let $B_1^\star=B_1\setminus B_1^\max$ and $\cB^\star$ the partition obtained from $\cB$ by replacing the block $B_1$ with the blocks $B_1^\max$ and $B_1^\star$.
			As $B_1^\max$ is a union of blocks of $\cA_{i-1}^P$ and $\cA_{i-1}^P\sqsubseteq \cB$, by construction, we have $\cA_{i-1}^P\sqsubseteq \cB^\star$.
			We claim that $(X,\cB^\star)$ is a good partition that admits $(x_1,\dots,x_t)$ as a good permutation, this is sufficient to prove the claim.
			
			As $\cc(G-X)\sqsubseteq \cA_{i-1}^P$, by definition we have $\cc(G-X)\sqsubseteq \cB^\star$.
			So Property~\ref{item:sqsubset} of Definition~\ref{def:goodpartition} is satisfied.
			As $B_1^\star$ and $B_1^\max$ are subsets of $B_1$, we deduce that $G[X\cup B_1^\star]$ and $G[X\cup B_1^\max]$ are both $X$-interval graphs.
			Thus, Property~\ref{item:Xinterval} is satisfied.
			
			It remains to prove that $(x_1,\dots,x_t)$ is a good permutation of $(X,\cB^\star)$.
			First, observe that for every $j\in[i-1]$, by definition, $x_j$ is removable from $X_j$ for $\cA_j^P$.
			Since $\cA_j^P\sqsubseteq \cA_{\ell-1}^P \sqsubseteq \cB^\star$, we deduce that $x_j$ is removable from $X_j$ for $\cB^\star$.
			
			Let $j\in \{i,\dots,t\}$.
			The vertex $x_j$ is not minimal in at most one block of $\cB$ for $X_j$.
			We claim that $N(x_j)$ is minimal in $B_1^\max$ for $X_j$.
			Since the partition $\cB^\star$ is obtained from $\cB$ by splitting the block $B_1$ into $B_1^\star$ and $B_1^\max$, this implies that $x_j$ is removable from $X_j$ for $\cB^\star$.
			We prove that $N(x_j)$ is minimal in $B_1^\max$ for $X_j$ by showing that, for every $x_\ell\in X_i$, we have $B_1^\max\subseteq N(x_\ell)$.
			
			Since $A_1\subseteq \notmin(w,X_i, \cA_{i-1}^P)$, we have $N(w)\cap A_1\neq \emptyset$ and thus $w\in \maxcap(A_1)$. 
			Let $A$ be a block of $\cA_{i-1}^P$ contained in $B_1^\max$.
			As $A_1\leq_X A$, we deduce that $w\in \mincap(A)$, that is $A\subseteq N(w)$.
			By definition, $A$ is not contained in $\notmin(w,X_i, \cA_{i-1}^P)$ and thus $N(w)$ is minimal in $A$ for $X_i$. Consequently, for every $x_\ell\in X_i$, we have $A\subseteq N(x_\ell)$.
			As this holds for every $A\subseteq B_1^\max$, we deduce that $B_1^\max \subseteq N(x_\ell)$ for every $x_\ell\in X_i$.
			This ends the proof of Claim~\ref{claim:Bmax}.
		\end{claimproof}
		
		From now, based on Claim~\ref{claim:Bmax}, we assume that $B_1^\max=\emptyset$.
		We construct $\cB_\new$ as follows. Recall that $A_1\subseteq B_1$. We create a new block $\widehat{B}_1=B_1\cup A_2\cup \dots \cup A_{k}$, and for every block $B\in \cB$ such that $B\neq B_1$, we create a new block $\widehat{B}= B \setminus (A_2\cup \dots\cup A_{k})$. We define $\cB_\new=\{\widehat{B}_1\}\cup \{\widehat{B}\mid B\in \cB\setminus\{B_1\} \text{ and } \widehat{B}\neq \emptyset\}$.
		The construction of $\cB_\new$ is illustrated in Figure~\ref{fig:Bnew}.
		
\begin{figure}[ht]
    \centering
	\includegraphics[width=0.95\linewidth]{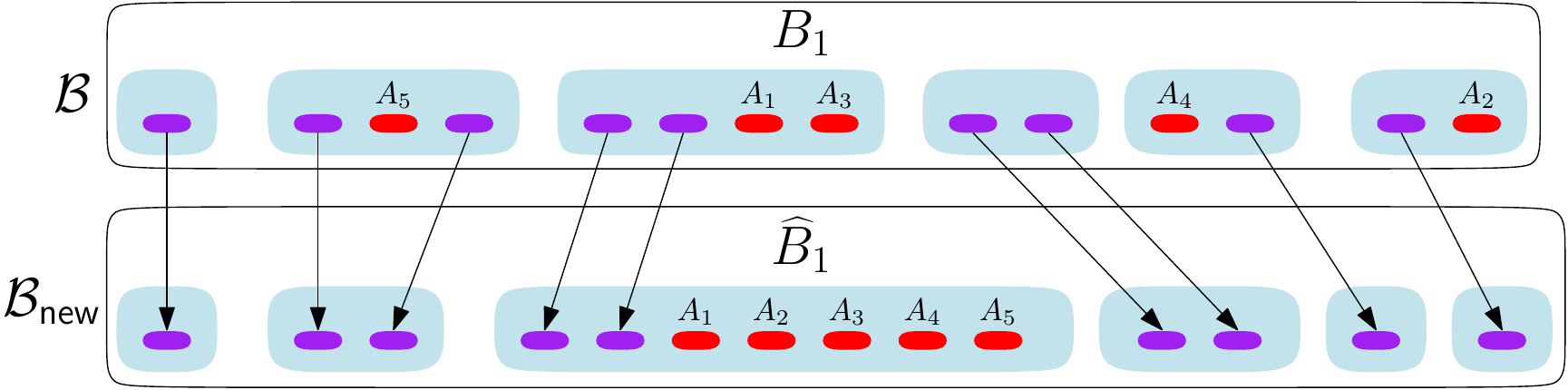}
    \caption{Construction of $\cB_\new$ from $\cB$ with $k=5$. The blocks of $\cB$ and $\cB_\new$ are in blue, the blocks $A_1,\dots,A_5$ of $\cA_{i-1}^P$ contained in $\notmin(w,X_i,\cA_{i-1}^P)$ are in red, the other blocks of $\cA_{i-1}^P$ are in purple. In each block of $\cB$ or $\cB_\new$, the blocks of $\cA_{i-1}^P$ are ordered w.r.t. $\leq_X$ from left to right.}
    \label{fig:Bnew}
\end{figure}
			
		We claim that $(X,\cB_\new)$ is a good partition that admits a good permutation starting with $(x_1,\dots,x_{i-1},w)$.
		By construction, $\cB_\new$ is a partition of $V(G)\setminus X$ and  $\cA_{i-1}^P \sqsubseteq \cB_\new$.
		As $\cc(G-X)\sqsubseteq \cA_{i-1}^P$ by definition, we have $\cc(G-X)\sqsubseteq \cB_\new$.
		So $(X,\cB_\new)$ satisfies Property~\ref{item:sqsubset} of Definition~\ref{def:goodpartition}.
		
		Observe that for every block $\widehat{B}$ of $\cB_\new$ such that $\widehat{B}\neq \widehat{B}_1$, $\widehat{B}$ is a subset of the block $B$ of $\cB$.
		Since $(X,\cB)$ is a good partition, $G[X\cup B]$ is an $X$-interval graph and hence $G[X\cup \widehat{B}]$ is also an $X$-interval graph.
		To prove that Property~\ref{item:Xinterval} is satisfied, it remains to show that $G[X\cup \widehat{B}_1]$ is an $X$-interval graph, for which we will use Claim~\ref{claim:leqX}.
		
		Since $G[X\cup B_1]$ is an $X$-interval graph we know by Claim~\ref{claim:comparable} that the blocks of $\cA_{i-1}^P$ contained in $B_1$ are pairwise comparable for $\leq_X$.
		Since $B_1^\max=\emptyset$, for every block $A\in \cA_{i-1}^P$ such that $A\subseteq B_1$ and $A$ is not contained in $\notmin(w,X_i,\cA_{i-1}^P)$, we have $A\leq_X A_1$.
		Moreover, we have $A_1\leq_X \dots \leq_X A_k$.
		We deduce that the blocks of $\cA_{i-1}^P$ contained in $\widehat{B}_1$ are pairwise comparable for $\leq_X$ because $\leq_X$ is transitive.
		By Claim~\ref{claim:leqX}, this implies that $G[X\cup \widehat{B}_1]$ is an $X$-interval graph.
		We conclude that Property~\ref{item:Xinterval} is satisfied.
		
		\medskip
		
		It remains to prove that $(X,\cB_\new)$ admits a good permutation starting with $(x_1,\dots,x_{i-1},w)$.
		We prove it with the following four claims.
		
		\begin{claim}\label{claim:start}
			For every $ j\in[i-1]$, $x_j$ is removable from $X_j$ for $\cB_\new$ and $w$ is removable from $X_i$ for $\cB_\new$.
			Moreover, $N(w)$ is minimal for $X_i$ in every block of $\cB_\new$ different from $\widehat{B}_1$.
		\end{claim}
		\begin{claimproof}
			For every $j\in [i-1]$, we have $\cA_j^P\sqsubseteq \cA_{i-1}^P$ and thus $x_j$ is removable from $X_j$ for $\cA_i^P$.
			By construction of $\cB_\new$, we have  $\cA_{i-1}^P \sqsubseteq \cB_\new$.
			We deduce that $x_j$ is removable from $X_j$ for $\cB_\new$ for every $j\in[i-1]$.
			By construction of $\cB_\new$, we have $\notmin(w,X_i, \cA_{i-1}^P)\subseteq \widehat{B}_1$.
			Hence, $N(w)$ is minimal for $X_i$ in every block of $\cB_\new$ different from $\widehat{B}_1$.
			It follows that $w$ is removable from $X_i$ for $\cB_\new$.			
		\end{claimproof}
		
		Let $p\in \{i,\dots,t\}$ such that $x_p=w$.
		
		\begin{claim}\label{claim:biggerThanP}
			For every $j,\ell\in[t]$ such that  $i \leq j <\ell$ and $p <\ell$, if $N(x_j)$ is minimal in $B_1$ for $X_j$, then $N(x_j)\cap \widehat{B}_1\subseteq N(x_\ell)$.
		\end{claim}
		\begin{claimproof}
			Let $j,\ell\in [t]$ such that  $i \leq j <\ell$ and $p <\ell$ and $N(x_j)$ is minimal in $B_1$ for $X_j$.
			Assume towards a contradiction that $N(x_j)\cap \widehat{B}_1\not\subseteq N(x_\ell)$.
			It follows that $N(x_\ell)\cap \widehat{B}_1\subset N(x_j)$ because $G[X\cup \widehat{B}_1]$ is an $X$-interval graph (these neighborhoods in $\widehat{B}_1$ are comparable).
			As $N(x_j)$ is minimal in $B_1$ for $X_j$ and $\ell\in X_i$, we have\footnote{In fact, we have $N(x_j)\cap B_1= N(x_\ell)\cap B_1$ since $B_1\subseteq \widehat{B}_1$ and  $N(x_\ell)\cap \widehat{B}_1\subset N(x_i)$.}  $N(x_j)\cap B_1 \subseteq N(x_\ell)$.
			So by construction of $\widehat{B}_1$, there exists a block $A\in \cA_{i-1}^P$ such that $A\subseteq \widehat{B}_1\setminus B_1$ and $N(x_\ell) \cap A \subset N(x_j)\cap A$.
			This implies that $x_\ell\notin\mincap(A)$.
			Since $A\subseteq \widehat{B}_1\setminus B_1 \subseteq \notmin(w,X_i,\cA_{i-1}^P)$ and $A\not\subseteq B_1$, we have $A_1\leq_X A$ and thus $\maxcap(A_1)\subseteq \mincap(A)$.
			We deduce that $x_\ell\notin \maxcap(A_1)$.
			
			By definition $A_1\subseteq \notmin(w,X_i,\cA_{i-1}^P)$ and thus $w\in \maxcap(A_1)\subseteq \mincap(A)$.
			Hence, we have $N(x_\ell)\cap A\subset N(w)\cap A$ and $N(x_\ell)\cap A_1\subset N(w)\cap A_1$.
			Since $p<\ell$ and $w=x_p$, we have $x_\ell\in X_p$.
			Consequently, $N(w)$ is not minimal in $A$ and $A_1$ for $X_p$.
			Let $B_A$ be the block of $\cB$ containing $A$.
			Since $A_1\subseteq B_1$ and $A\subseteq B_A$, $N(w)$ is not minimal in $B_1$ and $B_A$ for $X_p$.
			But $B_1$ and $B_A$ are two distinct blocks of $\cB$ because $A\not\subseteq B_1$, this is a contradiction with $w$ being removable from $X_p$ for $\cB$.
		\end{claimproof}
		
		\begin{claim}\label{claim:end}
			For every $j\in\{p+1,\dots,t\}$, the vertex $x_j$ is removable from $X_j$ for $\cB_\new$.
		\end{claim}
		\begin{claimproof}
			Let $j\in \{p+1,\dots,t\}$.
			The vertex $x_j$ is removable from $X_j$ from $\cB$.
			So there exists a block $B_{x_j}\in \cB$ such that $N(x_j)$ is minimal for $X_j$ in every block of $\cB$ different from $B_{x_j}$.
			By construction, every block $\widehat{B}$ of $\cB_\new$ is associated with a block $B$ of $\cB$ and if $B\neq B_1$, we have $B\subseteq \widehat{B}$.
			Thus, for every block $\widehat{B}\in \cB_\new$ different from $B_1$ and $B_{x_j}$, $N(x_j)$ is minimal in $\widehat{B}$ for $X_j$ because $\widehat{B}\subseteq B$ and $N(x_j)$ is minimal in $B$ for $X_j$.
			If $B_{x_j}=B_1$, then $x_j$ is removable from $X_j$ for $\cB_\new$.
			
			Now suppose that $B_{x_j}\neq B_1$.
			It follows that $N(x_j)$ is minimal for $X_j$ in $B_1$.
			As $p<j$, by Claim~\ref{claim:biggerThanP}, for every $\ell > j$, $N(x_j)\cap \widehat{B}_1\subseteq N(x_\ell)$.
			Thus, $N(x_j)$ is minimal in $\widehat{B}_1$ for $X_j$.
			Consequently, $N(x_j)$ is minimal in every block of $\cB_\new$ different from $\widehat{B}_{x_j}$.
			In both cases, $x_j$ is removable from $X_j$ for $\cB_\new$.
		\end{claimproof}
		
		It remains to deal with the vertices of $(x_1,\dots,x_t)$ between $x_i$ and $x_{p-1}$.
	
		\begin{claim}\label{claim:middle}
			There exists a permutation $(x^\star_{i+1},\dots,x^\star_p)$ of $\{x_i,\dots,x_{p-1}\}$ such that, for every $j\in \{i+1,\dots,p\}$, $x^\star_j$ is removable from $\{x^\star_j,\dots,x^\star_p\}\cup X_{p+1}$ for $\cB_\new$.
		\end{claim}
		\begin{claimproof}
		    Since $(x_1,\dots,x_t)$ is a good permutation of $(X,\cB)$, for every $j\in\{i,\dots,p-1\}$, there exists a block $B_{x_j}$ of $\cB$ such that $N(x_j)$ is minimal for $X_j$ in every block of $\cB$ different from $B_{x_j}$.
		    Let $j\in \{i,\dots,p-1\}$. We start by proving that 
		    ($\spadesuit$) $N(x_j)$ is minimal for $X_i$ in every block of $\cB_\new$ different from $\widehat{B}_{x_j}$ and $\widehat{B}_1$.
		    Take $\widehat{B}\in\cB_\new$ different from $\widehat{B}_{x_j}$ and $\widehat{B}_1$.
		    By construction, $\widehat{B}$ is associated with a block $B$ of $\cB$.
		    Since $\widehat{B}\neq \widehat{B}_{x_j}$, $N(x_j)$ is minimal in $B$ for $X_j$.
		    As $j<p$ and $w=x_p$, we have $w\in X_j$ and thus $N(x_j)\cap B \subseteq N(w)$.
		    By Claim~\ref{claim:start}, $N(w)$ is minimal for $X_i$ in every block of $\cB_\new$ different from $\widehat{B}_1$.
			In particular, $N(w)$ is minimal in $\widehat{B}$ for $X_i$.
			As $N(x_j)\cap B \subseteq N(w)$ and $\widehat{B}\subseteq B$, we deduce that $N(x_j)$ is also minimal in $\widehat{B}$ for $X_i$. This proves ($\spadesuit$).
		
			Since $G[X\cup \widehat{B}_1]$ is an $X$-interval graph, the neighborhoods of $X$ in $\widehat{B}_1$ are pairwise comparable for the inclusion. Thus, there exists a permutation $(x_{i+1}^\star,\dots,x_p^\star)$ of $\{x_i,\dots,x_{p-1}\}$ such that $N(x_{i+1}^\star)\cap \widehat{B}_1 \subseteq N(x_{i+2}^\star)\cap \widehat{B}_1 \subseteq \dots \subseteq N(x_p^\star)\cap \widehat{B}_1$.

			Let $j\in \{i+1,\dots,p\}$ and $X_j^\star= \{x^\star_j,\dots,x^\star_p\}\cup X_{p+1}$.
			We need to show that $x^\star_j$ is removable from $X_j^\star$ for $\cB_\new$.
		    As $X_j^\star$ is a subset of $X_i$, from ($\spadesuit$), we know that $N(x_j^\star)$ is minimal for $X_j^\star$ in every block of $\cB_\new$ different from $\widehat{B}_1$ and $\widehat{B}_{x_j^\star}$.
		    Thus, if $\widehat{B}_1= \widehat{B}_{x_j^\star}$, then $x_j^\star$ is removable from $X_j^\star$ for $\cB_\new$.
		    
		    Suppose that $\widehat{B}_1\neq \widehat{B}_{x_j^\star}$. 
		    Then $N(x_j^\star)$ is minimal in $\widehat{B}_{1}$ for $X_q$ with $q\in [t]$ such that $x_j^\star=x_q$. 
		    By Claim~\ref{claim:biggerThanP}, for every $x_\ell\in X_{p+1}$, we have $N(x_j^\star)\cap \widehat{B}_1\subseteq N(x_\ell)$.
		    Since $N(x_{j}^\star)\cap \widehat{B}_1 \subseteq N(x_{j+1}^\star)\cap \widehat{B}_1 \subseteq \dots \subseteq N(x_p^\star)\cap \widehat{B}_1$ and $X_j^\star= \{x^\star_j,\dots,x^\star_p\}\cup X_{p+1}$, $N(x_j)$ is minimal for $X_j^\star$ in $\widehat{B}_1$.
		    We deduce that $N(x_j^\star)$ is minimal for $X_j^\star$ in every block of $\cB_\new$ different from $\widehat{B}_{x_j^\star}$.
		    Thus, $x_j^\star$ is removable from $X_j^\star$ for $\cB_\new$.
		\end{claimproof}
		
	Let $(x^\star_{i+1},\dots,x^\star_p)$ be the permutation of $\{x_i,\dots,x_{p-1}\}$ satisfying the condition of Claim~\ref{claim:middle}.
	Take $P_\new=(x_1,\dots,x_{i-1},w,x^\star_{i+1},\dots,x^\star_p,x_{p+1},\dots,x_t)$, from Claims~\ref{claim:start},~\ref{claim:end} and~\ref{claim:middle}, we deduce that $P_\new$ is a good permutation of $(X,\cB_{\new})$.
	Hence, $(X,\cB_{\new})$ satisfies Property~\ref{item:elimination:order} of Definition~\ref{def:goodpartition} and $(X,\cB_{\new})$ is a good partition of $G$.
	As $P_\new$ starts with $(x_1,\dots,x_{i-1},w)$, this concludes the proof of Lemma~\ref{lem:noMistake}.
\end{proof}

\begin{lemma} \label{lemma:findsone}
    If $G$ has a good partition then Algorithm \ref{algo:goodpartition} returns one.
\end{lemma}

\begin{proof}
	Suppose $G$ admits a good partition with central clique $X$. We prove the following invariant holds at the end of the $i$-th iteration of the while loop for $X$.
	
	{\bf Invariant.} $G$ admits a good permutation starting with $w_1,...,w_i$.

	By assumption, $X$ admits a good permutation and thus the invariant holds before the algorithm starts the first iteration of the while loop.
	By induction, assume the invariant holds when Algorithm~\ref{algo:goodpartition} starts the $i$-th iteration of the while loop.
	Let $L=(w_1,\dots,w_{i-1})$ be the consecutive vertices chosen at Line~\ref{line:choosingw} before the start of iteration $i$ (observe that $L$ is empty when $i=1$).
	The invariant implies that there exists a good permutation $P=(w_1,\dots,w_{i-1},x_i,x_{i+1},...,x_t)$ of $G$ starting with $L$.
	By Lemma~\ref{lem:goodPermut}, the graph $G[X\cup \notmin(x_{i},\{x_{i},\dots,x_t\},\cA_{i-1}^P)]$ is an $X$-interval graph.
	Observe that $\cA^P_{i-1}$ and $\{x_i,\dots,x_t\}$ are the values of the variables $\cA$ and $W$ when Algorithm~\ref{algo:goodpartition} starts the $i$-th iteration.
	Thus, at the start of the $i$-th iteration, there exists a vertex $w_i$ such that $G[X\cup \notmin(w_i,W,\cA)]$ is an $X$-interval graph.
	Consequently, the algorithm does not return \textsf{no} at the $i$-th iteration and chooses a vertex $w_i\in \{x_i,\dots,x_t\}$ such that the graph $G[X\cup \notmin(w_i,\{x_i,\dots,x_t\},\cA^P_i)]$ is an $X$-interval graph.
	By Lemma~\ref{lem:noMistake}, $G$ admits a good permutation starting with $(w_1,\dots,w_{i-1},w_i)$.
	Thus, the invariant holds at the end of the $i$-th iteration.
	If at the end of the $i$-th iteration $W$ is empty, then  the while loop stops and the algorithm returns a pair $(X,\cA)$.
	Otherwise, the algorithm starts an $i+1$-st iteration and the invariant holds at the start of this new iteration.
	By induction,  the invariant holds at every step and Algorithm~\ref{algo:goodpartition} returns a pair $(X,\cA)$ which is a good partition by Lemma \ref{lem:goodOutput}. 
\end{proof}

\begin{theorem}\label{thm:algo1}
Algo. \ref{algo:goodpartition} decides in polynomial time if $G$ admits a star NeS model.
\end{theorem}
\begin{proof}
	By Theorem~\ref{thm:goodpartition} $G$ admits a star NeS model iff $G$ admits a good partition, and by Lemmata~\ref{lem:goodOutput} and \ref{lemma:findsone} Algorithm~\ref{algo:goodpartition} finds a good partition iff the input graph has a good partition.
	Let us argue for the runtime.
Checking that $G$ is chordal and finding the $O(n)$ maximal cliques can be done in polynomial time~\cite{HabibMPV00,RoseTL76}.
	Given $X,Y\subseteq V(G)$ we check whether $G[X\cup Y]$ is an $X$-interval graph, as follows.
	Take $G'$ the graph obtained from $G[X\cup Y]$ by adding $u$ and $v$ such that $N(u)=\{v\}$ and $N(v)=\{u\}\cup X$.
	It is easy to see that $G[X\cup Y]$ is an $X$-interval graph iff $G'$ is an interval graph, which can be checked in polynomial time.
\end{proof}

\section{Conclusion}

	The question if leaf powers can be recognized in polynomial time was raised over 20 years ago \cite{NishimuraRT02}, and this problem related to phylogenetic trees still remains wide open. We have shown that polynomial-time recognition can be achieved if the weighted leaf root is required to be a caterpillar (using the connection to blue-red interval graphs and co-TT graphs) or if the NeS model has a single large-degree branching. 
	
	We strongly believe our results could be combined to recognize in polynomial time the graphs with a leaf root whose internal vertices induce a subdivided star, that is a subdivided star with some leaves added to it (this generalizes linear leaf powers as the internal vertices of a caterpillar induce a path).
	%
	Let us instead sketch how the techniques in this paper could be generalized to handle the general case of leaf powers.
    Given a NeS model $\cN=(\cT,(T_v)_{v\in V(G)})$  of a graph $G$ with $\cT$ the embedding of a tree $T$, we define the \emph{topology} of $\cN$ as the pair $(T',(X_u)_{u\in V(T')})$ where $T'$ is the subtree induced by the internal nodes of $T$ and for each node $u$ of $T'$,
    we define $X_u=\{ v \in V(G):u \in T_v\}$, which will be a clique of $G$. For a star NeS model, $T'$ is a single node and we proved that (1)~if a graph admits a star NeS model, then it admits one where the only clique in the topology is maximal and (2)~given a topology with $T'$ a single node, we can decide in polynomial time whether the graph admits a star NeS model with this topology.
    For the general case it would suffice to solve the following problems:
	\begin{itemize}
	    \item Problem 1: given a graph $G$ compute in polynomial time a polynomial-sized family $\cF$ of topologies such that $G$ is a leaf power if and only if $G$ admits a NeS model with topology in $\cF$.
        A first step towards solving this problem would be to generalize Claim~\ref{claim:maximalclique} to show that we can always modify a NeS model so that all the cliques of its topology become maximal cliques (or minimal separators, as these are also manageable in chordal graphs). 
        A second step is to study the structural properties of strongly chordal graphs with the aim of showing a connection between these properties and a family $\cF$ of topologies of NeS models. 
        One candidate among these structural properties is the notion of \emph{clique arrangements} introduced by Nevries and Rosenke~\cite{NR15,NevriesR16}, which describes the intersections between the maximal cliques of a chordal graph $G$ more precisely than the clique trees. 
    	  
        \item Problem~2: given a graph and a topology, construct a NeS model with this topology or confirm that no such NeS model exists, in polynomial time.
        For star NeS models, this is handled inside the for loop of Algorithm~\ref{algo:goodpartition} and relies on Definition \ref{def:goodpartition}. A step towards solving this problem would be to generalize the combinatorial definition of a good partition to handle any topology.
	\end{itemize}
	We suggest to first employ this two-problems approach to graphs admitting NeS models with a very simple topology.
    Natural first cases are topologies $(T,(X_u)_{u\in V(T)})$ where either $T$ is an edge or where $\bigcup_{u\in V(T)} X_u$ is a clique. Both these restrictions generalize the topologies of star NeS models and positive results should produce tools useful to attack the general case.
    
    Let us end with a subproblem whose solution would simplify this approach towards settling the general case.
    Inspired by blue-red interval graphs, we can define blue-red NeS models where the red vertices 
    induce an independent set and for every red vertex $v$ its neighbors are the blue vertices $u$ such that $T_v\subseteq T_u$. 
     We can adapt the proof of Theorem~\ref{thm:equivalence} to show that $G$ admits a NeS model iff $G$ has a blue-red NeS model where the blue vertices are simplicial vertices without true twins,
   so these have the same modeling power.
    However, for blue-red NeS models the necessary topologies can be significantly simpler.
    For example, NeS models of linear leaf powers are caterpillars, whose topologies are paths, while its blue-red NeS models can be restricted to an edge, whose topology is empty.
    Can we generalize the concept of blue-red NeS models so that we can recursively color more vertices  blue (by allowing some of them to be adjacent under some conditions) and thereby simplify the topologies we need to consider to solve Problems 1 and 2?

\bibliography{biblio.bib}

\end{document}